\documentclass[final,3p,times]{elsarticle}

\usepackage{amssymb}
\usepackage{amsmath}
\usepackage{amsthm}
\newtheorem{proposition}{Proposition}
\newtheorem{corollary}{Corollary}
\newtheorem{theorem}{Theorem}[section]
\usepackage{microtype}
\usepackage{booktabs}
\usepackage{lineno}
\usepackage{graphicx}
\graphicspath{{../figures/}{./}}
\usepackage[hypertexnames=false]{hyperref}

\journal{TBD}

\begin{document}

\begin{frontmatter}

\title{DiLLSUE: a differentiable GPU solver for link-based logit stochastic user equilibrium}
\author[1]{Yue Li}
\author[1]{Shujuan Chen\corref{cor1}}
\ead{sc2343@cam.ac.uk}
\author[1]{Ying Jin}

\cortext[cor1]{Corresponding author.}

\affiliation[1]{organization={Martin Centre for Architectural and Urban Studies,
                University of Cambridge},
            addressline={1--5 Scroope Terrace},
            city={Cambridge},
            postcode={CB2 1PX},
            country={United Kingdom}}

\begin{abstract}
Logit stochastic user equilibrium (SUE) captures imperfect route cost
perceptions, but its practical solvers rely on route enumeration, which
becomes prohibitive on large networks, or on link-based heuristics without
convergence guarantees.
We develop DiLLSUE, a differentiable solver for the link-based logit SUE that
requires no route enumeration, no training data, and no network-specific
tuning.
Its inner loading algorithm batches all destinations into fixed-shape tensor
operations, yielding, to our knowledge, the first GPU implementation of
link-based logit SUE.
Four outer solvers --- successive averages, self-regulating averaging,
Anderson mixing, and implicit-function-theorem Newton --- are benchmarked
under identical settings, and a family of acyclicity filters provides fast
approximations with a quantified speed--accuracy trade-off.
On five standard benchmark networks, GPU and CPU executions agree to within
$10^{-4}$\% mean absolute percentage error, and the computed equilibrium
converges monotonically to the independently computed Wardrop equilibrium,
reaching 0.33\% error on Sioux Falls, where the zero-training solver is more
accurate than published trained surrogates.
\end{abstract}


\end{frontmatter}

\section{Introduction}
\label{sec:intro}
Reliable traffic models are central to transport planning, providing the
quantitative basis for infrastructure investment decisions, network design,
congestion pricing, and the assessment of air-quality, emissions, and public
health impacts~\cite{IPCC2022_Mitigation,airquality,emission,%
trafficvolumedecisionmaking,chen_parta}.
Transport planning outcomes depend directly on how traffic distributes across
a network in response to changes in travel demand and network infrastructure
supply.
Traffic assignment is the procedure that determines this distribution.
It finds a network equilibrium in which the distribution of travellers
across routes is consistent with their route choice behaviour.
Infrastructure appraisal, land-use impact assessment, congestion pricing
evaluation, and long-term demand forecasting all require evaluating how link
flows change under alternative demand conditions and network configurations.
In practice, no single scenario characterises the planning problem.
Sensitivity analyses, stochastic demand modelling, and policy portfolio
evaluation each require the assignment problem to be solved many times, and a
comprehensive uncertainty assessment can demand hundreds to thousands of
individual evaluations.
On networks representative of real planning models, iterative equilibrium
solvers require minutes to hours of computation per scenario, and a large-scale
evaluation task can consume days to weeks of aggregate computing
time~\cite{YedavalliEtAl2022,HuXie2025}.
This computational burden limits both the scope of policy alternatives that
practitioners can test and the resolution at which planning uncertainty can be
assessed.

Traffic equilibrium models differ in how they represent route choice behaviour.
Wardrop user equilibrium (UE) assumes that all travellers have perfect route
information and select the fastest available route, with flows concentrating
on minimum-cost paths until no traveller can improve their travel time
unilaterally~\cite{Wardrop1952,Beckmann1956}.
This assumption is appropriate in intra-urban contexts where frequent commuters
develop accurate cost perceptions through repeated
experience~\cite{Sheffi1985}.
Logit stochastic user equilibrium (SUE) replaces perfect information with
multinomial logit route choice~\cite{Dial1971,BenAkiva1985}, capturing the
heterogeneous and imperfect cost perceptions observed in practice, particularly
on inter-urban corridors where route familiarity and information levels are
lower~\cite{DaganzoSheffi1977,Fisk1980}.
SUE is behaviourally richer, but existing solvers rely on explicit route
enumeration, which becomes computationally prohibitive on large
networks~\cite{Prashker2004,Sheffi1985}.

Two approaches have been developed to address logit SUE in practice.
The first uses path-based solvers operating on explicit route sets.
Column generation builds the route set iteratively by adding improving
shortest-path columns until the assignment converges to the exact Fisk
logit equilibrium~\cite{Fisk1980,Sheffi1985}. It is the current exact
path-based standard but requires repeated all-pairs shortest-path computation
whose cost scales with network size and iteration count.
Top-$K$ path selection and Monte Carlo route sampling reduce this cost at the
price of approximation error relative to the full column-generation
solution~\cite{Prashker2004}.
All these strategies maintain OD-pair-specific route sets of variable length,
forming a structural barrier to simultaneous evaluation across destinations.
The second trains neural networks to approximate equilibrium flow patterns.
Graph attention networks (GAT) trained on thousands of pre-solved OD scenarios
can predict link flows with competitive accuracy and reduced inference
cost~\cite{HuXie2025,VelickovicEtAl2018}.
Earlier graph-based and implicit learning approaches have extended this
direction~\cite{LiuMeidani2024,LiuEtAl2023,LiuYin2025,LiEtAl2026a,LiEtAl2026b}.

Both approaches face fundamental limitations.
Path-based logit SUE solvers face a structural scaling barrier.
Column generation requires an all-pairs shortest-path solve per outer
iteration. On networks with large OD pair counts, a single iteration can take
hours of computation.
Route sampling and top-$K$ selection reduce this cost but introduce
approximation error that is difficult to bound without the full
column-generation solution as a reference.
In either case, the path-set structure imposes a ceiling that is intrinsic to
the route-enumeration formulation.
Surrogate approaches face a different set of constraints.
Training data must be generated by the solver the surrogate is designed to
replace, and a separate model must be fitted for each network before any
inference cost savings are realised. When fewer scenarios are needed than the
training budget, the reference solver is cheaper overall.
Replacing a mechanistic glass-box equilibrium model with a neural network also
loses causal transparency, as the relationship between demand inputs and predicted
flows is absorbed into network weights, making the model difficult to audit,
interpret, or extend to policy scenarios outside the training distribution.
GAT surrogates also might not enforce flow conservation at intermediate
nodes, violating a fundamental continuity constraint of traffic
assignment~\cite{HuXie2025}.

Link-based logit SUE avoid route enumeration entirely, reducing the complexity of logit SUE substantially to tractable range.
Rather than distributing demand over an explicitly generated route set, they
model route choice as a sequence of link-level decisions at successive
nodes, loading demand directly onto links so that no route is ever stored.
The dominant practical tool, Dial's STOCH~\cite{Dial1971}, approximates logit
assignment through link efficiency weights without explicit route storage,
but provides no convergence guarantee and no gap measure.
Akamatsu~\cite{Akamatsu1996} formulated
logit assignment on the full cyclic network as an absorbing Markov chain,
Baillon and Cominetti~\cite{BaillonCominetti2008} established existence and
uniqueness of the resulting Markovian traffic equilibrium, recursive
logit~\cite{FosgerauEtAl2013} gave the structure modern random utility
foundations, and network generalized extreme value (NGEV) extensions
followed~\cite{OyamaEtAl2022}.

Three decades on, however, this rigorous line has not become a practical
solver, for the reasons that this paper addresses in turn.
The first is computation. The equilibrium carries a value function and a
flow state for every destination, so implementations tracking
OD-specific flows scale quadratically with the number of
zones in memory and per-iteration work (terabytes of state on a
metropolitan network) while even destination-wise implementations sweep
one destination at a time on CPU.
Applications have therefore remained small-scale case
studies~\cite{Akamatsu1996,BaillonCominetti2008,OyamaEtAl2022}.
The second is the outer solver. Each study adopts a single fixed-point scheme, and
no comparison establishes which methods converge, over which dispersion
range, at what cost per iteration.
The third is well-posedness. The exact cyclic model can fail to admit a
fixed point and can assign flow to circular walks, and the acyclic
restrictions that repair this carry accuracy costs that have never been
quantified.
Finally, no computed equilibrium in this line has been verified externally.
Every study checks only its own convergence residual, never the solution
itself against an independent implementation or against an established
equilibrium computed by a different method.

This study makes three contributions, organised in
Figure~\ref{fig:framework}.
First, we develop DiLLSUE (differentiable link-based logit SUE), a solver
for the exact link-based logit SUE.
The inner algorithm core is destination batching, which reduces the flow
state from $O(|\mathcal{N}|\,|D|^2)$ to $O(|\mathcal{N}|\,|D|)$ (from
terabytes to hundreds of megabytes on metropolitan networks) and
collapses $|D|$ sequential Bellman sweeps into one tensor kernel per pass,
yielding the first GPU implementation of link-based logit SUE.
Second, we build four outer solvers for the equilibrium fixed
point. Successive averages~\cite{SheffiPowell1982}, self-regulating
averaging~\cite{LiuEtAl2009}, Anderson mixing~\cite{Walker2011}, and
implicit-function-theorem Newton.
Third, we develop fast approximations of DiLLSUE through a family of
acyclicity filters and quantify the trade-off between speed and accuracy
against the exact full-graph equilibrium.
We validated our algorithms on five standard benchmark networks,
comparing accuracy, computational cost, and convergence state under identical
settings.

\begin{figure}[t]
\centering
\includegraphics[width=\textwidth]{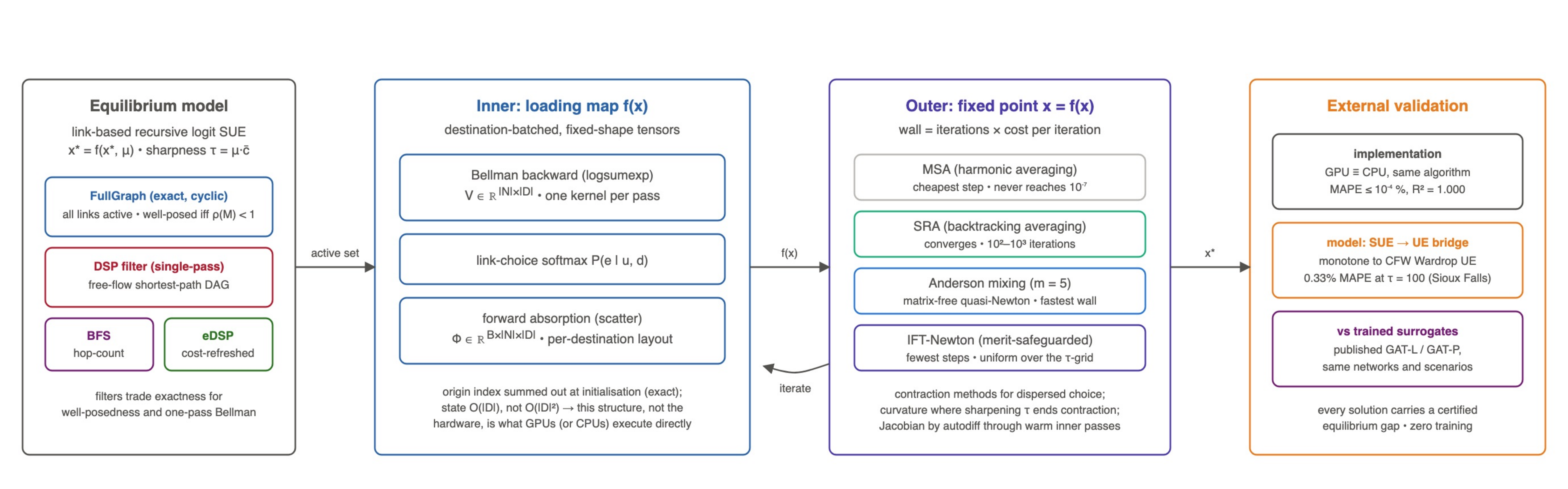}
\caption{The DiLLSUE framework: the paper's three contributions and
their validation.
An equilibrium model layer selects the exact cyclic full-graph model or an
acyclicity filter (DSP, BFS, eDSP), which trade exactness for
well-posedness and a single-pass Bellman.
The inner layer evaluates the loading map $\mathbf{f}(\mathbf{x})$ with all
destinations batched into fixed-shape tensors. This algorithmic structure,
not the hardware, is what allows execution as a small number of accelerator
kernels on GPU or CPU.
The outer layer solves the fixed point $\mathbf{x} = \mathbf{f}(\mathbf{x})$
and is where wall time decomposes into iterations times cost per iteration.
Four solvers are evaluated on the identical loading map.
The validation layer checks the result externally at three levels:
implementation (GPU against CPU), model (monotone convergence to the
independently computed Wardrop UE), and practice (accuracy against published
trained surrogates on identical inputs).}
\label{fig:framework}
\end{figure}

\section{Methods}
\label{sec:methods}

\subsection{Problem statement}
\label{sec:notation}
Let $\mathcal{G} = (\mathcal{N}, \mathcal{E})$ be a directed network with
node set $\mathcal{N}$ and link set $\mathcal{E}$.
Each link $e = (u,v) \in \mathcal{E}$ has free-flow travel time $t^0_e > 0$
and capacity $c_e > 0$.
Travel time on link $e$ depends on flow $x_e$ through the Bureau of Public
Roads (BPR) cost function~\cite{BPR1964}:
\begin{equation}
t_e(x_e) = t^0_e \!\left[1 + a_e\!\left(\frac{x_e}{c_e}\right)^{\!b_e}\right],
\label{eq:bpr}
\end{equation}
where $a_e > 0$ and $b_e > 0$ are link-specific BPR parameters.
The OD demand matrix $\mathbf{Q} = \{q_{od}\}$ gives the vehicle flow per hour
from origin $o$ to destination $d$.
The set of origins is $O \subseteq \mathcal{N}$, the set of destinations is
$D \subseteq \mathcal{N}$ with cardinality $|D|$, and $\mathcal{OD} \subseteq O \times D$
denotes the set of non-zero OD pairs with cardinality $|\mathcal{OD}|$.
The link flow vector $\mathbf{x} = (x_e)_{e \in \mathcal{E}}$ is the primary
output of traffic assignment.

Two equilibrium concepts are used throughout, illustrated in
Figure~\ref{fig:ue_sue}.
Wardrop user equilibrium (UE) requires that no traveller can reduce travel time
by unilaterally switching routes~\cite{Wardrop1952,Beckmann1956}:
\begin{equation}
\mathbf{x}^* = \mathbf{f}_{\mathrm{UE}}(\mathbf{x}^*),
\label{eq:ue}
\end{equation}
where $\mathbf{f}_{\mathrm{UE}}$ is an all-or-nothing loading map that assigns
all demand to minimum-cost paths under costs $t_e(\mathbf{x}^*)$.

Logit stochastic user equilibrium (SUE) under the recursive logit
model~\cite{FosgerauEtAl2013} allows travellers to make imperfect route choices
with heterogeneity governed by the dispersion parameter $\mu > 0$.
A value function $V(n,d)$ encodes the expected cost from node $n$ to destination
$d$ under logit route choice. The logit SUE fixed point is
\begin{equation}
\mathbf{x}^* = \mathbf{f}_{\mathrm{SUE}}(\mathbf{x}^*,\,\mu),
\label{eq:sue}
\end{equation}
where $\mathbf{f}_{\mathrm{SUE}}$ assigns demand to links via recursive logit
probabilities derived from $V(n,d)$ at costs $t_e(\mathbf{x}^*)$.
As $\mu \to +\infty$, logit route choice concentrates on minimum-cost paths
and logit SUE converges to Wardrop UE.
At small $\mu$, demand spreads more uniformly across routes.

\begin{figure}[t]
\centering
\includegraphics[width=\textwidth]{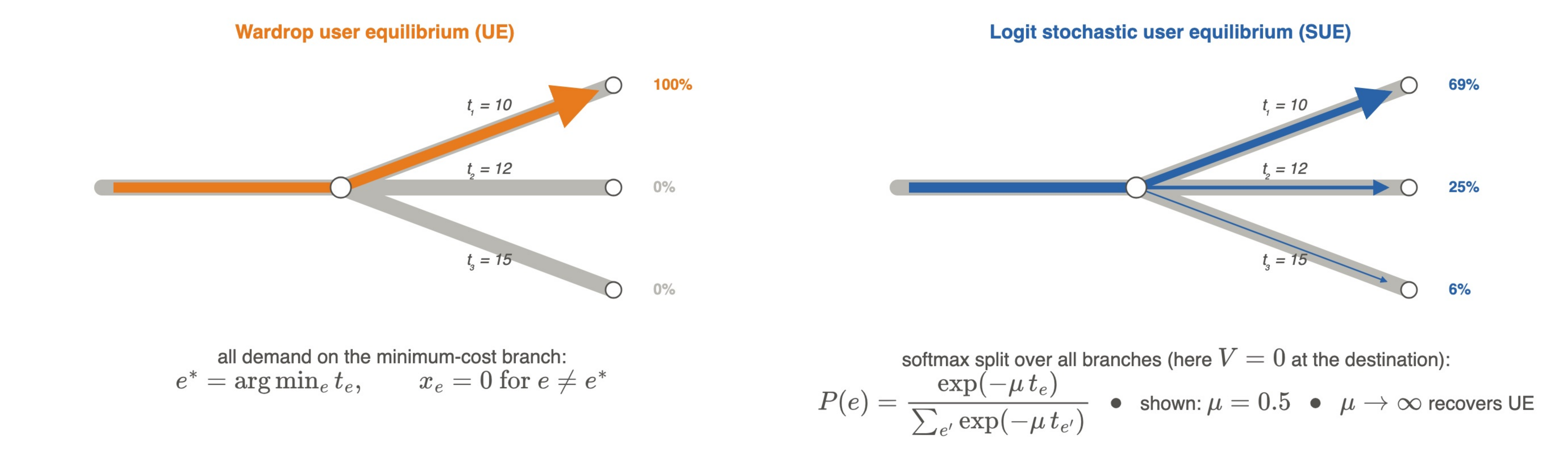}
\caption{Wardrop UE and logit SUE at a junction with three downstream
branches to the destination (costs 10, 12, 15).
UE assigns all demand to the minimum-cost branch.
Logit SUE splits demand across all branches by the softmax
$P(e) = \exp(-\mu t_e) / \sum_{e'} \exp(-\mu t_{e'})$, shown at
$\mu = 0.5$ (the value function $V$ vanishes here because every branch ends
at the destination).
As $\mu \to \infty$ the split concentrates on the minimum-cost branch and
SUE recovers UE.}
\label{fig:ue_sue}
\end{figure}

\subsection{DiLLSUE Formulation}
\label{sec:dita}

Following differentiable traffic assignment (DiTA)~\cite{NieLi2026}, which
recasts the iterative assignment process as a differentiable computation
graph executable by modern tensor libraries, DiLLSUE separates the solution
of equation~\eqref{eq:sue} into two components.
A loading map $\mathbf{f}$ evaluates the recursive logit assignment at the
current costs (Section~\ref{sec:loading}).
An outer solver drives the flows to the fixed point
$\mathbf{x}^* = \mathbf{f}(\mathbf{x}^*)$
(Section~\ref{sec:outer}), where $\mathbf{f}(\mathbf{x}_\ell)$ denotes
loading at costs $t_e(\mathbf{x}_\ell)$ and $\ell = 0, 1, 2, \ldots$
indexes the outer iteration.

\subsubsection{Recursive logit loading map}
\label{sec:loading}
In DiLLSUE, the loading map $\mathbf{f}$ is a logit assignment computed by
recursive logit~\cite{FosgerauEtAl2013}.
For each destination $d \in D$, a backward value iteration computes the
value function $V(n,d)$ satisfying the Bellman equation over all outgoing
links at each node:
\begin{equation}
V(n,\,d) = \log \sum_{e=(n,v) \in \mathcal{E}}
\exp\,\bigl[-\mu\,t_e(\mathbf{x}) + V(v,\,d)\bigr],
\label{eq:bellman}
\end{equation}
with $V(d,d) = 0$.
Under positive link costs and $\mu > 0$, the value iteration converges to a
unique fixed point~\cite{FosgerauEtAl2013}. It is run for $L_{\mathrm{val}}$
passes.
Figure~\ref{fig:bellman} traces these sweeps on a toy network.
The logit choice probability for link $e = (u,v)$ departing node $u$ toward
destination $d$ is:
\begin{equation}
P(e \mid u,\,d) = \exp\,\bigl[-\mu\,t_e(\mathbf{x}) + V(v,\,d) - V(u,\,d)\bigr].
\label{eq:logit}
\end{equation}
A forward absorption pass then propagates OD demand across the network:
\begin{equation}
\phi^d_n = \sum_{o \in O} q_{od}\,\mathbf{1}_{[n=o]}
         + \sum_{e=(u,n) \in \mathcal{E}}
           P(e \mid u,\,d)\;\phi^d_u,
\qquad n \neq d,
\label{eq:absorption}
\end{equation}
where $\phi^d_n$ is the flow mass at node $n$ en route to destination $d$,
initialised to $q_{od}$ at each origin $o$.
The loading map output is
$f_e(\mathbf{x}) = \sum_{d \in D} P(e \mid u, d)\;\phi^d_u$
for each link $e = (u,v)$.

\begin{figure}[t]
\centering
\includegraphics[width=\textwidth]{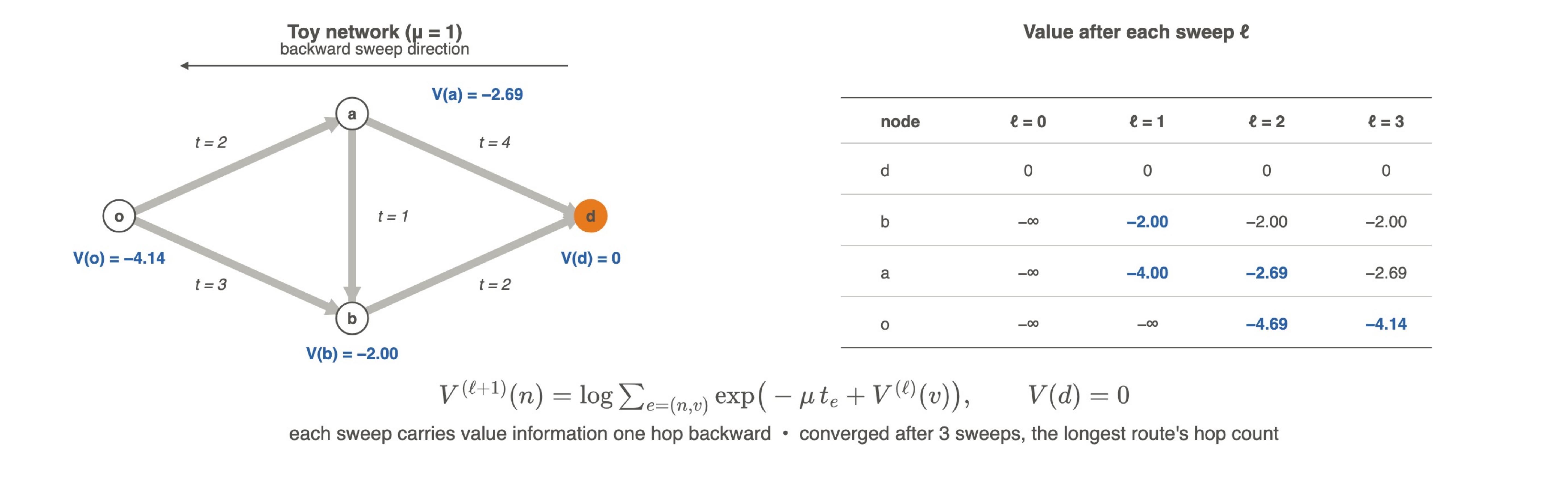}
\caption{The backward Bellman sweep on a toy network with a single
destination $d$ ($\mu = 1$, link costs as shown).
Each synchronous sweep applies equation~\eqref{eq:bellman} once and carries
value information one hop backward from the destination.
After sweep 1 only nodes one hop from $d$ hold finite values.
Node $a$ improves at sweep 2, when the cheaper two-hop route through $b$
becomes visible, and the origin converges at sweep 3, the hop count of the
longest route.
$L_{\mathrm{val}}$ sweeps therefore suffice once $L_{\mathrm{val}}$ covers
the network's hop diameter.}
\label{fig:bellman}
\end{figure}

\subsubsection{Outer solver}
\label{sec:outer}
Write $\mathbf{r}_\ell = \mathbf{f}(\mathbf{x}_\ell) - \mathbf{x}_\ell$
for the fixed-point residual at outer iteration $\ell$, so that the
equilibrium satisfies $\mathbf{r}(\mathbf{x}^*) = \mathbf{0}$.
Four outer solvers are built on the identical loading map.
They differ only in how the next iterate is formed from the residual, and
they span the trade-off between many cheap steps and few expensive ones.

\paragraph{Method of successive averages}
The textbook baseline~\cite{SheffiPowell1982} averages each loading into
the current flows with the harmonic step $\alpha_\ell = 1/(\ell+1)$,
starting from $\mathbf{x}_0 = \mathbf{0}$:
\begin{equation}
\mathbf{x}_{\ell+1}
  = \mathbf{x}_\ell + \alpha_\ell\,\mathbf{r}_\ell
  = \left(1 - \alpha_\ell\right)\mathbf{x}_\ell
    + \alpha_\ell\,\mathbf{f}(\mathbf{x}_\ell).
\label{eq:msa}
\end{equation}
The harmonic step satisfies the two Robbins-Monro summability conditions,
$\sum_{\ell} \alpha_\ell = \infty$ and
$\sum_{\ell} \alpha_\ell^2 < \infty$, which guarantee convergence to
the fixed point for the SUE loading
map~\cite{RobbinsMonro1951,SheffiPowell1982,Cantarella1997}.
The guarantee is bought with speed: the step shrinks on schedule whether or
not progress is being made.

\paragraph{Self-regulating averaging}
Self-regulating averaging (SRA)~\cite{LiuEtAl2009} keeps the averaging
direction $\mathbf{r}_\ell$ but sizes the step by backtracking on the
squared-residual merit
$W(\mathbf{x}) = \|\mathbf{f}(\mathbf{x}) - \mathbf{x}\|_2^2$:
starting from $\alpha = 1$, the step is halved until
$W(\mathbf{x}_\ell + \alpha\,\mathbf{r}_\ell) < W(\mathbf{x}_\ell)$,
and the harmonic step $\alpha_\ell = 1/\ell$ is taken when no halving
succeeds within eight attempts.
The fallback preserves the Robbins-Monro guarantee, while accepted
backtracking steps are far larger than harmonic steps away from the
equilibrium.

\paragraph{Anderson mixing}
Anderson acceleration~\cite{Walker2011} extrapolates the next iterate from
a short history of the $m_\ell = \min(m, \ell)$ most recent steps.
Collect the recent differences as columns of
$\Delta X_\ell = [\,\mathbf{x}_{\ell-m_\ell+1}-\mathbf{x}_{\ell-m_\ell},
\ldots, \mathbf{x}_{\ell}-\mathbf{x}_{\ell-1}\,]$ and
$\Delta R_\ell = [\,\mathbf{r}_{\ell-m_\ell+1}-\mathbf{r}_{\ell-m_\ell},
\ldots, \mathbf{r}_{\ell}-\mathbf{r}_{\ell-1}\,]$, solve the small
least-squares problem
\begin{equation}
\boldsymbol{\gamma}_\ell
  = \arg\min_{\boldsymbol{\gamma} \in \mathbb{R}^{m_\ell}}
    \bigl\| \mathbf{r}_\ell - \Delta R_\ell\,\boldsymbol{\gamma} \bigr\|_2,
\label{eq:anderson_ls}
\end{equation}
and form the type-II update
\begin{equation}
\mathbf{x}_{\ell+1}
  = \mathbf{x}_\ell + \mathbf{r}_\ell
    - \left(\Delta X_\ell + \Delta R_\ell\right)\boldsymbol{\gamma}_\ell.
\label{eq:anderson}
\end{equation}
Equation~\eqref{eq:anderson} is a multisecant quasi-Newton step: it
approximates the Newton direction from the residual history alone, at the
cost of one $m_\ell$-dimensional least-squares solve per iteration and with
no Jacobian evaluations.
The candidate is accepted only when it reduces the merit $W$, and the solver
falls back to the SRA step otherwise, so the safeguarded iteration reaches
the same fixed point and is never worse than its fallback.
The production setting uses window $m = 5$ with Tikhonov-regularised normal
equations.
To our knowledge this is the first application of Anderson mixing to traffic
assignment.
A convergence certificate for a variant safeguarded by sufficient decrease
of the entropy-regularised SUE objective is given in Supplementary
Section~S8, together with the gap between that variant and the cheaper
residual-based safeguard used in production.

\paragraph{Implicit-function-theorem Newton}
The production solver applies Newton's method to the residual equation
$\mathbf{r}(\mathbf{x}) = \mathbf{0}$.
With $J_\ell = \partial \mathbf{f} / \partial \mathbf{x}$ evaluated at
$\mathbf{x}_\ell$, each step solves the linear system
\begin{equation}
\left(I - J_\ell\right)\boldsymbol{\delta}_\ell = \mathbf{r}_\ell,
\qquad
\mathbf{x}_{\ell+1} = \mathbf{x}_\ell + \boldsymbol{\delta}_\ell.
\label{eq:newton}
\end{equation}
Two obstacles make the textbook form impractical here.
The Jacobian is a dense $|\mathcal{E}| \times |\mathcal{E}|$ matrix that
cannot be formed, and $\mathbf{f}$ is itself defined through the inner
Bellman fixed point of equation~\eqref{eq:bellman}.
The first is resolved by solving equation~\eqref{eq:newton} matrix-free
with restarted GMRES, which is valid for the asymmetric operator
$I - J_\ell$ and requires only Jacobian--vector products.
The second is resolved through the implicit-function theorem: because the
inner value function is converged before each step, the products
$J_\ell \mathbf{v}$ are evaluated by forward-mode automatic
differentiation through a fixed number of Bellman passes initialised at the
warm-converged value function, which differentiates the converged inner
operator without storing or unrolling the inner iteration history.
This construction is what the differentiable implementation supplies
natively, and it is new to this setting.
A short SRA phase warm-starts the flows into Newton's basin of attraction,
and every Newton step is safeguarded by a merit test on $W$ with an MSA
fallback.

All four solvers run unchanged on GPU and CPU.
Their four-way comparison at identical tolerance is reported in
Supplementary Section~S1.

\subsubsection{Acyclicity filter variants}
\label{sec:filter}
The full-graph Bellman backward requires $L_{\mathrm{val}}$ iterative passes,
each of cost $O(|\mathcal{E}| \times |D|)$, making it the dominant cost per
outer iteration.
An acyclicity filter reduces this to a single pass by restricting the Bellman
sum in equation~\eqref{eq:bellman} to a per-destination DAG, on which
topological order enables exact propagation in one backward sweep.

The destination shortest-path (DSP) filter constructs the active link set as
\begin{equation}
\mathcal{E}^d_{\mathrm{DSP}} = \bigl\{(u,v)\in\mathcal{E} :
  d_{\mathrm{FF}}(v,\,d) < d_{\mathrm{FF}}(u,\,d)\bigr\},
\label{eq:dsp}
\end{equation}
where $d_{\mathrm{FF}}(n,d)$ is the free-flow shortest-path distance from $n$
to $d$. This set forms a DAG per destination (Supplementary Material,
Section~S2).
The filter is computed once from base free-flow costs and reused across all
outer iterations, introducing an approximation gap because equilibrium costs
shift link preferences away from the free-flow shortest-path tree. The gap is
quantified in Section~\ref{sec:results_filter}.

Two acyclicity filter variants are evaluated in the Supplementary Material.
A BFS hop-count filter restricts the active link set using shortest-hop counts
from the free-flow network, offering a coarser but computationally cheaper
acyclicity constraint (Section~S3).
The equilibrium-cost DSP (eDSP) rebuilds the filter at each outer iteration
from the current equilibrium costs, progressively reducing the approximation
gap without requiring a full Bellman backward (Section~S4).

\subsection{Destination-batched GPU computation}
\label{sec:gpu}

The loading map $\mathbf{f}(\mathbf{x})$ at each outer iteration comprises two
operations: the Bellman backward pass, which computes $V(n,d)$
via equation~\eqref{eq:bellman}, and the forward absorption pass, which
propagates OD demand via equation~\eqref{eq:absorption}.
Both operations are independent across destinations and are batched over all
$|D|$ destinations simultaneously on the GPU.

\paragraph{Bellman backward}
The value function $V(n,d)$ depends only on the destination $d$.
All $|D|$ Bellman sweeps are therefore stacked into a single value matrix
$\mathbf{V} \in \mathbb{R}^{|\mathcal{N}| \times |D|}$. Each pass executes
as one scatter-logsumexp kernel call over the
$(|\mathcal{E}| \times |D|)$ link-destination message tensor, replacing $|D|$
sequential sweeps with one GPU kernel.

\paragraph{Forward absorption}
The same destination-indexed structure applies to the forward absorption pass.
Let $\boldsymbol{\Phi} \in \mathbb{R}^{B \times |\mathcal{N}| \times |D|}$
denote the batched per-destination flow tensor, where $B$ is the scenario
batch dimension and $\Phi_{b,n,d}$ is the flow mass at node $n$ en route to
destination $d$ in scenario $b$ (the stacked form of $\phi^d_n$ from
equation~\eqref{eq:absorption}).
An alternative per-OD layout
$\boldsymbol{\Phi}^{\mathrm{OD}} \in
\mathbb{R}^{B \times |\mathcal{N}| \times |\mathcal{OD}|}$
tracks each origin--destination pair separately.
For a dense OD matrix, $|\mathcal{OD}| \approx |D|^2$, so the per-OD layout
scales quadratically in the number of destinations while the per-destination
layout scales linearly. The reduction factor is $|\mathcal{OD}|/|D|$.
The reformulation is exact: the link-choice
probability~\eqref{eq:logit} depends on the destination alone, so
destination-indexed propagation reproduces the per-OD flows by summing demand
over origins at initialisation rather than tracking each origin separately.
Supplementary Section~S9 states this equivalence and the exact
$(|D|-1)$-fold reduction as a proposition with proof.
Table~\ref{tab:layout} summarises the memory and per-pass work for each layout.
On Barcelona ($|D| = 110$, $|\mathcal{OD}| = 7{,}922$), the per-destination
layout reduces tensor state and per-pass scatter work by a factor of $72$
relative to the per-OD layout.
The factor is what makes larger networks representable at all: on a
metropolitan network of the scale of Sydney (33{,}113 nodes, 3{,}264 zones,
TransportationNetworks repository), the per-OD flow tensor for a single
scenario would occupy roughly 1.4\,TB in single precision, beyond any single
accelerator, while the per-destination tensor occupies roughly 432\,MB.

\begin{table}[ht]
\centering
\caption{Memory and per-pass scatter work for the two forward absorption layouts.
         $B$: scenario batch size. $|\mathcal{N}|$: nodes.
         $|\mathcal{E}|$: links. $|D|$: destinations.
         $|\mathcal{OD}|$: OD pairs.
         For a dense OD matrix $|\mathcal{OD}| \approx |D|^2$, making the
         reduction factor $|\mathcal{OD}|/|D| \approx |D|$.}
\label{tab:layout}
\begin{tabular}{lll}
\toprule
Layout & $\boldsymbol{\Phi}$ tensor state
       & Per-pass scatter work \\
\midrule
Per-OD          & $B \times |\mathcal{N}| \times |\mathcal{OD}|$
                & $B \times |\mathcal{E}| \times |\mathcal{OD}|$ \\
Per-destination & $B \times |\mathcal{N}| \times |D|$
                & $B \times |\mathcal{E}| \times |D|$ \\
Reduction       & $|\mathcal{OD}|/|D|$
                & $|\mathcal{OD}|/|D|$ \\
\bottomrule
\end{tabular}
\end{table}

This architecture is enabled by the link-level formulation of recursive logit:
because route choice is encoded through $V(n,d)$ rather than enumerated
explicitly, the value tensor $\mathbf{V}$ has the same fixed shape
$|\mathcal{N}| \times |D|$ regardless of route diversity.
Destination batching is structurally distinct from origin-parallel GPU
approaches in deterministic assignment~\cite{HeywoodEtAl2019} and is
unavailable to path-based SUE solvers, which maintain OD-pair-specific route
sets of variable length and cannot be arranged into fixed-shape tensors.

\subsection{Experimental design and setup}
\label{sec:setup}

\paragraph{Networks}
Five benchmark networks are used, drawn from the open TransportationNetworks
repository~\cite{Stabler2023}, the same five on which Hu and
Xie~\cite{HuXie2025} trained and evaluated their GAT model, enabling direct
comparison on identical inputs.
Table~\ref{tab:networks} summarises the network characteristics.

\begin{table}[ht]
\centering
\caption{Benchmark networks from the TransportationNetworks repository.
         Zones are OD-generating areas (traffic analysis zones).
         OD pairs counts non-zero entries in the base demand matrix.
         Volume-to-capacity ratio ($v/c$) is computed under base demand.
         $\bar{c}$ is the demand-weighted mean OD free-flow shortest-path
         cost, used to normalise the dispersion sweep
         (equation~\eqref{eq:tau}).}
\label{tab:networks}
\begin{tabular}{lrrrrcc}
\toprule
Network & Nodes & Links & Zones & OD pairs & $v/c^{*}$ & $\bar{c}$ \\
\midrule
Sioux Falls           &    24 &    76 &  24 &    528 & $\approx 1.47$ & 8.81 \\
Eastern Massachusetts &    74 &   258 &  74 &  1{,}113 & $\approx 0.30$ & 0.38 \\
Anaheim               &   416 &   914 &  38 &  1{,}406 & $\approx 0.60$ & 11.17 \\
Barcelona             & 1{,}020 & 2{,}522 & 110 &  7{,}922 & $\approx 0.10$ & 6.50 \\
Winnipeg              & 1{,}052 & 2{,}836 & 147 &  4{,}345 & $\approx 0.10$ & 12.24 \\
\bottomrule
\end{tabular}
\par\smallskip
\footnotesize $^{*}$Barcelona and Winnipeg use uniform unit link capacities in
the TNTP source files. Their $v/c$ values reflect average flow scale relative
to that unit capacity and are not directly comparable with the other networks.
Under the perturbed demand these two networks operate in the saturated regime
of the capped BPR function (see below).
\end{table}

The five networks span link count from 76 to 2{,}836, OD pair count from 528
to 7{,}922, and loading level from lightly loaded to heavily congested
(volume-to-capacity ratio $v/c$ from 0.1 to 1.47).
Sioux Falls is the standard small-network test case and the most congested
network in the benchmark set.
Barcelona and Winnipeg are the two largest networks in the set.
Their TNTP files assign uniform unit link capacities, so the perturbed
demand drives volume-to-capacity ratios far above one.
The ratio $x_e/c_e$ is therefore capped at $100$ before the power in
equation~\eqref{eq:bpr} is applied, a saturated-link modelling assumption
applied identically in every solver in this study.
The cap is inactive on the other three networks.

\paragraph{Dimensionless dispersion sweep}
The behavioural sharpness of logit route choice is governed not by $\mu$
alone but by the product of $\mu$ with the travel costs it multiplies in
equation~\eqref{eq:bellman}.
The same $\mu$ therefore produces very different choice behaviour on networks
whose cost scales differ: mean OD free-flow shortest-path costs span a factor
of about thirty across the five benchmark networks
(Table~\ref{tab:networks}).
A raw $\mu$ grid applied uniformly across networks would compare
near-deterministic assignment on one network with near-uniform spreading on
another at the same nominal grid point.
The sweep is therefore indexed by the dimensionless sharpness
\begin{equation}
\tau = \mu\,\bar{c},
\label{eq:tau}
\end{equation}
where $\bar{c}$ is the demand-weighted mean OD free-flow shortest-path cost of
the network: the typical cost against which a traveller compares alternatives,
computable a priori without solving any equilibrium.
Demand weighting reflects that trip distribution is non-uniform. Free-flow
rather than congested cost avoids circularity with the equilibrium being
solved.
Each network is evaluated at $\tau \in \{1,\,3,\,10,\,30,\,100\}$, covering
strongly dispersed route choice ($\tau = 1$) through the near-deterministic
regime ($\tau = 100$) in which logit SUE approaches Wardrop UE. The
corresponding raw $\mu = \tau/\bar{c}$ per network follows from
Table~\ref{tab:networks}.

\paragraph{Scenario generation}
Two hundred scenarios are generated per network by independently perturbing each
element of the OD demand matrix, link capacities, and free-flow travel times:
\begin{align}
\tilde{q}_{od}  &= \sigma_{od}\,q_{od},
  & \sigma_{od}  &\sim U(0.5,\;3.0),
\label{eq:perturb_q}\\
\tilde{c}_e     &= \sigma_{c,e}\,c_e,
  & \sigma_{c,e} &\sim \mathrm{DU}\{0.5,\;1.0,\;2.0,\;3.0\},
\label{eq:perturb_c}\\
\tilde{t}^0_e   &= \sigma_{t,e}\,t^0_e,
  & \sigma_{t,e} &\sim \mathrm{DU}\{0.8,\;1.0,\;1.2,\;1.4\},
\label{eq:perturb_t}
\end{align}
where $U$ denotes the continuous uniform distribution, $\mathrm{DU}$ denotes
the discrete uniform distribution, and all multipliers are drawn independently
per element (seed 42, first 200 of 2{,}000 scenarios).
This protocol follows Hu and Xie~\cite{HuXie2025} exactly, enabling direct
comparison of performance metrics.

\paragraph{DiLLSUE solvers}
Three DiLLSUE configurations are evaluated.
DiLLSUE-FullGraph uses the full-graph Bellman loading map with the IFT-Newton
outer solver (Section~\ref{sec:dita}).
DiLLSUE-DSP uses the DSP acyclicity filter for a single-pass Bellman at each
outer iteration with the same outer solver.
DiLLSUE-BFS and DiLLSUE-eDSP results are reported in the Supplementary
Material.
All DiLLSUE runs use $L_{\mathrm{val}} = 50$ Bellman passes per outer
iteration, sufficient to cover the maximum network diameter (Winnipeg: 40 hops),
run in double precision, and declare convergence at the L2 fixed-point
relative residual
\begin{equation}
\mathrm{gap}_{\mathrm{rel}} =
  \frac{\|\mathbf{x} - \mathbf{f}(c(\mathbf{x}))\|_2}{\|\mathbf{x}\|_2}
  < 10^{-7}.
\label{eq:gap}
\end{equation}
The threshold $10^{-7}$ is motivated by the basin-of-attraction behaviour of
Barcelona and Winnipeg at looser tolerances (Supplementary Material,
Section~S5).
The cyclic full-graph model admits a fixed point only when the spectral radius
of its exponentiated-utility weight matrix is below one. FullGraph results are
reported for the cells in which this fixed point exists and the solver reached
tolerance for all 200 scenarios, which covers Sioux Falls at all $\tau$,
Eastern Massachusetts at $\tau \geq 10$, and Barcelona and Winnipeg at all
$\tau$ (Section~\ref{sec:limitations} discusses the remaining cells).
On the two saturated unit-capacity networks the converged full-graph flows
contain substantial cyclic-walk traversals (Section~\ref{sec:limitations}),
so comparisons against the path-based Wardrop reference on Barcelona and
Winnipeg are reported with the acyclic DSP variant. The full-graph cells
there enter the implementation-consistency comparison only.
The acyclic DSP variant is well-posed on every network and every $\tau$.

\paragraph{Reference solvers}
Ref-SUE-FullGraph runs the identical IFT-Newton algorithm on CPU.
Comparison~1 reports per-link flow agreement between the GPU and CPU
executions: an implementation-consistency check that hardware, kernel
scheduling, and floating-point reduction order do not change the computed
fixed point, not a measure of flow quality against any external planning
benchmark.
Model correctness is instead established by Comparison~4: as $\tau$ grows the
computed logit SUE must converge to the Wardrop UE computed by an entirely
independent solver family. A defective loading map would not exhibit this
limit.
Wardrop UE reference flows were generated by the conjugate Frank-Wolfe
algorithm~\cite{MitradjievaLindberg2013} at Wardrop relative gap below
$10^{-5}$~\cite{HuXie2025} under the same capped BPR cost.

\paragraph{GAT surrogates}
We use published results from Hu and Xie~\cite{HuXie2025} directly and
do not train new surrogate models.
Hu and Xie trained two GAT variants on Wardrop UE ground truth generated by
CFW: GAT-L (link-based) and GAT-P (path-based).
MAPE and $R^2$ are taken from Table~3 of Hu and Xie~\cite{HuXie2025}.

\paragraph{Comparisons}
Four result sections address the paper's objectives.
Comparison~1 (Section~\ref{sec:results_sue_acc}) reports per-link flow agreement
between GPU and CPU executions of DiLLSUE-FullGraph across the $\tau$ grid:
the implementation-consistency check.
Comparison~2 (Section~\ref{sec:results_sue_gat}) evaluates DiLLSUE against
published GAT-L and GAT-P results from Hu and Xie~\cite{HuXie2025},
both measured against the CFW Wardrop UE reference on identical networks and
scenarios.
Comparison~3 (Section~\ref{sec:results_compute}) reports GPU against CPU
wall-clock time for the same algorithm and convergence criterion.
Comparison~4 (Section~\ref{sec:bridge}) traces the computed SUE toward the
independently computed Wardrop UE across the full $\tau$ sweep: the model
correctness check.
The DSP filter's speed-accuracy trade-off relative to FullGraph is quantified
in Section~\ref{sec:results_filter}.

\begin{table}[ht]
\centering
\caption{Summary of experimental configurations.
         All DiLLSUE and Ref-SUE-FullGraph runs use $L_{\mathrm{val}} = 50$
         inner Bellman passes per outer iteration, double precision, and the
         IFT-Newton outer solver.
         $\mathrm{gap}_{\mathrm{rel}}$: equation~\eqref{eq:gap}.
         Wardrop rel.\ gap: standard CFW criterion~\cite{MitradjievaLindberg2013}.
         GAT hardware and training details follow Hu and Xie~\cite{HuXie2025}.}
\label{tab:config}
\begin{tabular}{lllll}
\toprule
Method & Hardware & Criterion & Threshold & $\tau$ \\
\midrule
DiLLSUE-FullGraph & A100 SXM4 (80\,GB)  & $\mathrm{gap}_{\mathrm{rel}}$
  & $10^{-7}$ & 1--100 \\
DiLLSUE-DSP       & A100 SXM4 (80\,GB)  & $\mathrm{gap}_{\mathrm{rel}}$
  & $10^{-7}$ & 1--100 \\
DiLLSUE-BFS$^{*}$ & A100 SXM4 (80\,GB)  & $\mathrm{gap}_{\mathrm{rel}}$
  & $10^{-7}$ & 10 \\
DiLLSUE-eDSP$^{*}$ & A100 SXM4 (80\,GB) & $\mathrm{gap}_{\mathrm{rel}}$
  & $10^{-7}$ & 1--100 \\
Ref-SUE-FullGraph  & Xeon 8360Y (CPU)    & $\mathrm{gap}_{\mathrm{rel}}$
  & $10^{-7}$ & 1--100 \\
CFW (UE ref.)      & Xeon 8360Y (CPU)    & Wardrop rel.\ gap
  & $10^{-5}$ & Wardrop UE \\
GAT-L / GAT-P      & per~\cite{HuXie2025} & —
  & — & Wardrop UE \\
\bottomrule
\end{tabular}
\par\smallskip\footnotesize
$^{*}$Supplementary Material only.
\end{table}

\paragraph{Metrics}
Two metrics follow Hu and Xie~\cite{HuXie2025}.
Mean absolute percentage error (MAPE) is:
\begin{equation}
\mathrm{MAPE} = \frac{1}{|\mathcal{E}^+|} \sum_{e \in \mathcal{E}^+}
\frac{|\hat{x}_e - x_e|}{x_e} \times 100\%,
\label{eq:mape}
\end{equation}
where $\hat{x}_e$ is the predicted flow, $x_e$ is the reference flow, and
$\mathcal{E}^+ = \{e \in \mathcal{E} : x_e \geq 1\;\mathrm{veh/h}\}$ excludes
near-zero-flow links to avoid denominator inflation.
The coefficient of determination is:
\begin{equation}
R^2 = 1 - \frac{\displaystyle\sum_{e \in \mathcal{E}} (\hat{x}_e - x_e)^2}
               {\displaystyle\sum_{e \in \mathcal{E}} (x_e - \bar{x})^2},
\label{eq:r2}
\end{equation}
where $\bar{x}$ is the mean reference flow across all links.
MAPE is the primary metric. $R^2$ measures the fraction of variance in
reference flows explained by predicted flows.
All metrics are computed per scenario and averaged across the 200 scenarios
per network.

\paragraph{Hardware and software}
GPU runs used an NVIDIA A100 SXM4 (80\,GB) on the CSD3 Ampere partition
(Cambridge HPC).
CPU runs used the CSD3 icelake partition: Intel Xeon Platinum~8360Y
(2$\times$36-core, Ice Lake, 2.4\,GHz), 32~cores per task, 120\,GB RAM.
Software: Python~3.10, PyTorch~2.x, CUDA~12.x (GPU). Python~3.10, NumPy,
SciPy (CPU).

\section{Results}
\label{sec:results}


\subsection{Implementation validity: DiLLSUE-FullGraph vs Ref-SUE-FullGraph}
\label{sec:results_sue_acc}

Table~\ref{tab:dita_sue_acc} reports agreement between the GPU and CPU
executions of DiLLSUE-FullGraph across the $\tau$ grid, 200 scenarios per
cell.
MAPE here measures agreement between two hardware executions of the same
recursive logit SUE algorithm and is an implementation-consistency check, not
a measure of equilibrium concept divergence.

\begin{table}[ht]
\centering
\caption{DiLLSUE-FullGraph GPU--CPU implementation consistency.
         Both executions run the identical IFT-Newton algorithm to
         $\mathrm{gap}_{\mathrm{rel}} < 10^{-7}$. The reported MAPE isolates
         hardware, kernel-scheduling, and floating-point reduction-order
         differences at the converged equilibrium.
         Worst cell = largest per-cell mean MAPE across the $\tau$ grid
         (200 scenarios per cell).}
\label{tab:dita_sue_acc}
\begin{tabular}{lccc}
\toprule
Network               & $\tau$ cells & Worst-cell MAPE\,(\%) & $R^2$ \\
\midrule
Sioux Falls           & 1--100  & $3.6\times10^{-7}$ & 1.000000 \\
Eastern Massachusetts & 10--100 & $1.3\times10^{-7}$ & 1.000000 \\
Barcelona             & 1--100  & $6.1\times10^{-5}$ & 1.000000 \\
Winnipeg              & 1--100  & $3.3\times10^{-7}$ & 1.000000 \\
\bottomrule
\end{tabular}
\par\smallskip
\footnotesize Anaheim and the low-$\tau$ Eastern Massachusetts cells are
outside the reported FullGraph scope (Section~\ref{sec:setup}). Anaheim is
covered by the DSP variant throughout.
\end{table}

Median MAPE is zero in every cell: the two executions produce bitwise-equal
flows for most scenarios, and the worst per-cell mean across eighteen
network-$\tau$ cells is $6.1\times10^{-5}$\,\% (Barcelona, $\tau = 1$).
Agreement at the floating-point floor confirms that the GPU-batched and CPU
executions compute the same recursive logit SUE fixed point. Nothing in the
destination-batched kernel path changes the solution.

\subsection{Comparison with trained GAT surrogate}
\label{sec:results_sue_gat}

Hu and Xie~\cite{HuXie2025} train two GAT variants on Wardrop UE ground
truth generated by CFW on the same five benchmark networks and the same
scenario protocol used here.
Table~\ref{tab:gat_benchmark} reports DiLLSUE accuracy at the
near-deterministic grid point $\tau = 100$ against the same CFW Wardrop UE
reference, alongside published GAT-L and GAT-P results from Hu and
Xie~\cite{HuXie2025}.

Note on model comparability: DiLLSUE targets the logit SUE fixed point,
while GAT surrogates approximate the Wardrop UE fixed point.
Both are evaluated against the same CFW reference, so the DiLLSUE MAPE values
include the residual equilibrium concept gap between logit SUE at $\tau = 100$
and Wardrop UE (traced across the full $\tau$ sweep in
Section~\ref{sec:bridge}). GAT MAPE is purely surrogate error against UE
ground truth.
The comparison is informative because practitioners face the same choice ---
equilibrium model and solver --- and both methods are evaluated on identical
inputs without any network-specific tuning.

\begin{table}[ht]
\centering
\caption{DiLLSUE at $\tau = 100$ against the CFW Wardrop UE reference,
         compared with GAT-L and GAT-P from Hu and Xie~\cite{HuXie2025}
         (their Table~3) on the same five networks and scenario protocol.
         DiLLSUE MAPE includes the residual SUE--UE equilibrium concept gap at
         $\tau = 100$. On the saturated unit-capacity networks (Barcelona,
         Winnipeg) and under the DSP filter this concept-plus-filter gap
         dominates.
         GAT results: mean across three or more training seeds from Hu and
         Xie~\cite{HuXie2025}. DiLLSUE requires no training.}
\label{tab:gat_benchmark}
\begin{tabular}{lrrrrrrrr}
\toprule
& \multicolumn{2}{c}{DiLLSUE-FullGraph} &
  \multicolumn{2}{c}{DiLLSUE-DSP} &
  \multicolumn{2}{c}{GAT-L} &
  \multicolumn{2}{c}{GAT-P} \\
\cmidrule(lr){2-3}\cmidrule(lr){4-5}\cmidrule(lr){6-7}\cmidrule(lr){8-9}
Network & MAPE & $R^2$ & MAPE & $R^2$ & MAPE & $R^2$ & MAPE & $R^2$ \\
\midrule
Sioux Falls           & \textbf{0.33} & 0.999 & 24.1 & 0.705 & 6.6 & 0.973 & 8.4 & 0.963 \\
Eastern Massachusetts & 11.2 & 0.998 & 22.1 & 0.904 & 8.1 & 0.975 & \textbf{4.7} & 0.979 \\
Anaheim               & --   & --    & 72.9 & 0.879 & 8.5 & 0.980 & \textbf{7.0} & 0.991 \\
Barcelona             & --   & --    & 65.0 & 0.962 & \textbf{4.6} & 0.991 & 6.7 & 0.989 \\
Winnipeg              & --   & --    & 73.7 & 0.932 & 4.9 & 0.990 & \textbf{3.8} & 0.995 \\
\bottomrule
\end{tabular}
\par\smallskip
\footnotesize MAPE in \%. FullGraph is reported within its well-posed scope
(Section~\ref{sec:setup}).
\end{table}

Where the exact cyclic model reaches its deterministic limit, the
zero-training solver is more accurate than the trained surrogates: 0.33\%
MAPE on Sioux Falls against 6.6--8.4\% for GAT, at $R^2 = 0.999$.
On Eastern Massachusetts the FullGraph residual (11.2\% mean, 4.7\% median)
is comparable to the trained surrogates.
On the remaining networks the reported DSP variant carries both the filter
approximation and the SUE--UE concept gap, and the trained surrogates ---
which are fitted directly to UE labels --- remain closer to the UE reference.
The structural trade-off, developed in Section~\ref{sec:bridge} and the
Discussion, is that DiLLSUE solves an exact behavioural equilibrium with a
certified convergence gap and no training, whereas the surrogate reproduces
the labels it was trained on within their distribution.

DiLLSUE-FullGraph requires no training, generates no upfront label cost, and
applies to any network from its description alone.
The GAT training pipeline requires running the reference solver at least
1{,}200 times to generate labelled flows before any inference is possible.
On Sioux Falls, the break-even threshold --- the scenario count at which the
full GAT pipeline cost equals the cost of running CFW directly --- is
approximately 2{,}890 scenarios. On Barcelona it is approximately 2{,}133
scenarios~\cite{HuXie2025}.
Below these thresholds, running the reference solver directly is cheaper.
DiLLSUE-FullGraph has no break-even threshold.


\subsection{Compute time}
\label{sec:results_compute}

Table~\ref{tab:compute} reports median per-scenario wall time at $\tau = 10$
under DiLLSUE-FullGraph on GPU (NVIDIA A100) and on CPU (Intel Xeon 8360Y,
32 cores), running the identical IFT-Newton algorithm to the same
$\mathrm{gap}_{\mathrm{rel}} < 10^{-7}$ criterion.
Because algorithm and convergence behaviour are identical, the wall-clock
ratio measures the hardware contribution of destination batching alone.

\begin{table}[ht]
\centering
\caption{Median per-scenario wall time at $\tau = 10$,
         $\mathrm{gap}_{\mathrm{rel}} < 10^{-7}$, identical IFT-Newton
         algorithm on both platforms.
         GPU: NVIDIA A100-SXM4 80\,GB (CSD3 Ampere).
         CPU: Intel Xeon Platinum 8360Y (CSD3 icelake), 32 cores.
         Ratio = CPU / GPU.}
\label{tab:compute}
\begin{tabular}{lrrr}
\toprule
Network & GPU (s) & CPU (s) & Ratio \\
\midrule
Sioux Falls           & 41.7 & 16.4 & 0.4 \\
Eastern Massachusetts & 34.2 & 42.0 & 1.2 \\
Barcelona             & 82.3 & 247.1 & 3.0 \\
Winnipeg              & 85.5 & 168.3 & 2.0 \\
\bottomrule
\end{tabular}
\end{table}

The ratio grows with network size, from 0.4 on Sioux Falls, where kernel
launch overhead dominates the tiny $24 \times 24$ value tensor and the CPU is
faster, to 3.0 on Barcelona.
This is the expected signature of destination batching: the GPU evaluates all
$|D|$ Bellman backward sweeps simultaneously as operations on the
$|\mathcal{N}| \times |D|$ value tensor (Section~\ref{sec:gpu}), so its
advantage scales with the amount of per-pass parallel work.
On Barcelona ($|D| = 110$ destination zones), a single Bellman pass on CPU
requires 110 sweeps. On the A100, all 110 are computed in one kernel call.
Two further reserves of GPU parallelism remain unexploited in these
measurements: scenarios are solved one at a time (the Newton outer step is
per-scenario), and the benchmark networks are one to two orders of magnitude
smaller than planning-scale models, where the per-destination tensor grows
into the regime the hardware is designed for.
The destination-batching architecture itself is structurally unavailable to
path-based SUE solvers regardless of implementation
effort~\cite{HeywoodEtAl2019}.

\subsection{DSP filter as a secondary approximation}
\label{sec:results_filter}

The DSP acyclicity filter (Section~\ref{sec:filter}) restricts the Bellman
backward sweep to free-flow efficient links, converting the full-graph
iteration to a single topological pass at the cost of a modelling
approximation.
Measured directly against the converged FullGraph equilibrium on the same
scenarios, the DSP approximation on Sioux Falls is a stable plateau of
approximately 24\% mean MAPE across the entire $\tau$ grid: the free-flow
shortest-path DAG excludes links that carry flow at the congested equilibrium,
and this exclusion does not diminish as route choice sharpens.
On Eastern Massachusetts the DSP gap is 26\% median MAPE at $\tau = 10$,
falling to 15\% at $\tau = 100$.

The equilibrium-cost refresh (eDSP) rebuilds the DAG from current equilibrium
costs at each outer iteration, and on Eastern Massachusetts it closes most of
the remaining gap: median MAPE against FullGraph falls from 21.5\% (DSP-level)
at $\tau = 10$ to 4.7\% at $\tau = 30$ and 2.9\% at $\tau = 100$.
The refresh recovers exactly the links whose congested costs make them
efficient even though their free-flow costs do not.
Filter choice is therefore not a neutral implementation detail. DSP is
recommended as the default single-pass filter because Section~S2 in the
Supplementary Material guarantees a connected DAG for every destination, a
property the coarser BFS hop-count filter does not provide.
BFS results are in Supplementary Section~S3, and full eDSP results with
their stability envelope are in Supplementary Section~S4.

\subsection{Convergence to Wardrop user equilibrium}
\label{sec:bridge}

Figure~\ref{fig:bridge} traces the computed logit SUE toward the
independently computed Wardrop UE across the full $\tau$ sweep on every
network.
As $\tau$ grows, MAPE against the CFW reference falls monotonically on every
network for both variants shown: on Sioux Falls, FullGraph falls from 28.6\% at
$\tau = 1$ to 0.33\% at $\tau = 100$, crossing below the best trained GAT
surrogate at $\tau \approx 3$. On Eastern Massachusetts it falls from 86\% at
$\tau = 10$ to 11.2\% at $\tau = 100$.
This monotone approach to a limit computed by an entirely different solver
family (convex combinations on the Beckmann program, no shared code, no
shared fixed-point structure) is the model-correctness anchor of the paper: a
defective loading map, filter, or outer solver would not reproduce the
theoretical $\tau \to \infty$ limit of the logit SUE model.

\begin{figure}[ht]
\centering
\includegraphics[width=\textwidth]{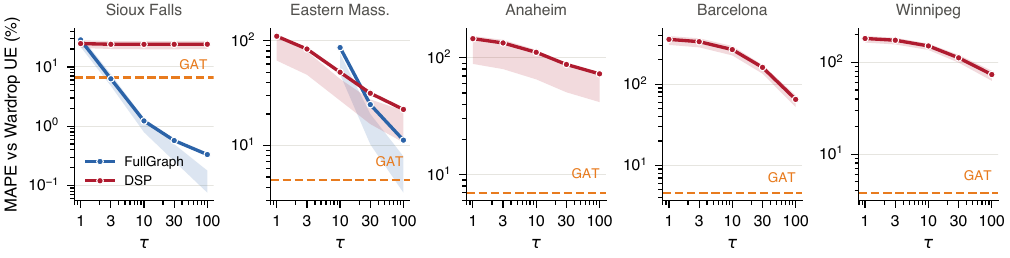}
\caption{SUE $\to$ UE bridge. MAPE of converged DiLLSUE link flows against
the CFW Wardrop UE reference as a function of the dimensionless sharpness
$\tau$, per network (mean over converged scenarios, ribbon: interquartile
range, 200 scenarios per cell).
FullGraph (blue) is shown within its well-posed scope
(Section~\ref{sec:setup}). DSP (red) on all five networks.
Dashed line: best published trained-GAT MAPE on that network (Hu \&
Xie~2025).
On Sioux Falls the zero-training exact solver crosses below the trained
surrogate at $\tau \approx 3$.
The DSP curves flatten toward their filter-approximation ceilings
(Section~\ref{sec:results_filter}). On the two saturated unit-capacity
networks (Barcelona, Winnipeg) the SUE--UE concept gap remains large at
$\tau = 100$ because the capped cost function decouples costs from flows over
much of the network.}
\label{fig:bridge}
\end{figure}

The bridge also separates the three error sources that the single-number GAT
comparison of Table~\ref{tab:gat_benchmark} conflates.
Implementation error is bounded by Table~\ref{tab:dita_sue_acc} at the
floating-point floor.
The filter approximation is the vertical offset between the DSP curve and the
FullGraph curve where both exist.
The remainder is the genuine behavioural difference between logit SUE at
finite $\tau$ and deterministic Wardrop assignment, which is a modelling
choice, not an error to be minimised.

\section{Discussion}
\label{sec:discussion}

DiLLSUE is, to our knowledge, the first GPU implementation of link-based logit
SUE assignment, requiring no route enumeration, no training data, and no
network-specific tuning.
It combines destination-batched recursive logit Bellman loading --- the
algorithmic structure that makes accelerator execution possible --- with an
IFT-Newton outer solver. The GPU and CPU executions agree at
the floating-point floor, and the computed equilibrium converges monotonically
to the independently computed Wardrop UE as dispersion vanishes.
On the most congested benchmark network the zero-training solver is more
accurate in that limit than the published trained GAT surrogates evaluated
on the same inputs.
The following sections discuss the conceptual and practical significance of
these results.

\subsection*{Why link-level recursive logit enables scalable SUE}

Existing SUE methods operate in path space.
Path-based MSA, Dial's STOCH~\cite{Dial1971}, and the path-based DiTA
formulation~\cite{NieLi2026} all enumerate routes explicitly.
Route enumeration scales at least exponentially with network size: even the
moderate Sioux Falls network (24 nodes, 76 links) admits 1{,}632{,}820
distinct simple paths across 528 OD pairs.
On larger networks, complete enumeration is computationally intractable.
Recursive logit~\cite{FosgerauEtAl2013} reformulates the logit SUE loading
step as a backward Bellman sweep on a directed acyclic subgraph, computing
exact logit choice probabilities without enumerating routes.
The computational cost scales with the number of links and destinations, not
with the number of paths.
This is the change that makes link-level logit SUE tractable on planning-scale
networks.

The acyclicity filter is the enabling structural requirement.
The Bellman equation~\eqref{eq:bellman} is well-defined only on a DAG.
On cyclic road networks, a filter must project the link set onto a DAG before
the backward sweep can proceed.
Section~S2 establishes that the DSP filter produces a
connected DAG for every destination, guaranteeing that the Bellman sweep is
both well-defined and that every reachable OD pair has at least one feasible
path in the filtered subgraph.
No prior work constructs or validates an acyclicity filter for recursive logit
in the forward assignment setting.
The filter design question was never forced until recursive logit was applied
to standard benchmark assignment networks without observed GPS paths to
constrain the active link set.

\subsection*{Outer solver and GPU destination-batching}

DiLLSUE separates two design questions: how many outer iterations the
equilibrium requires, and how much each iteration costs.

On the first, the four-way outer-solver comparison on Sioux Falls
(Supplementary Section~S1) shows that plain MSA fails to reach
$\mathrm{gap}_{\mathrm{rel}} < 10^{-7}$ within a 10{,}000-iteration ceiling,
adaptive averaging (SRA)~\cite{LiuEtAl2009} converges in hundreds to
thousands of iterations, Anderson mixing~\cite{Walker2011} converges in
roughly nine hundred cheap iterations and is fastest in wall time on that
network, and the production IFT-Newton reaches tolerance in one to two
hundred outer steps.
Algorithm comparisons for logit SUE date to Maher~\cite{Maher1998}, and
Newton-type methods have been developed for the path-based
formulation, most recently in a concurrent spectral analysis of the logit
mapping~\cite{BagchiBoyles2026}. The comparison here is, to our knowledge,
the first for the cyclic recursive-logit fixed point, the first application
of Anderson mixing to traffic assignment, and the first Newton method whose
Jacobian is obtained by automatic differentiation through the inner value
iteration.
Newton's advantage is robustness across the whole $\tau$ grid: its
merit-safeguarded steps converge uniformly through the near-deterministic
regime in which the fixed-point map loses contraction and
extrapolation-based accelerators degrade.
The pattern has a single mechanism.
The averaging iteration contracts while route choice is dispersed. As $\tau$
grows and choice sharpens toward the deterministic shortest path, the loading
map's contraction expires, extrapolation methods run out of usable descent
directions, and only a curvature method retains a convergence mechanism.
The same calculus, run the other way, explains the inner loop: the Bellman
value operator is a strong contraction on these networks, so warm-started
power iteration reaches the value fixed point in a few tens of cheap passes
and beats any curvature method there.
Curvature is worth its cost only where contraction is weak, which is why the
solver applies Newton to the outer problem and plain iteration to the inner
one.
The production choice of IFT-Newton trades some wall time on easy cells for
uniform convergence over every reported cell.

On the second, the GPU destination-batching architecture stacks the inner
Bellman and absorption sweeps for all $|D|$ destinations into fixed-shape
tensor operations: a single scatter-logsumexp over the
$(|\mathcal{E}| \times |D|)$ link-destination message tensor replaces $|D|$
sequential CPU sweeps.
GPU acceleration of the all-or-nothing loading step has been demonstrated for
deterministic assignment via origin-parallel multi-source Bellman-Ford
(parallelising over source nodes)~\cite{HeywoodEtAl2019}.
The architecture here parallelises over destination nodes and is enabled by
the fixed-shape value tensor of recursive logit: path-based methods maintain
OD-pair-specific route sets of variable length and cannot be organised into a
fixed-shape tensor regardless of implementation effort.
The wall-clock ratios of Table~\ref{tab:compute} grow with network size
exactly as this structural argument predicts, and the per-destination layout
reduces tensor state by a factor of $|\mathcal{OD}|/|D|$ relative to a per-OD
layout (Table~\ref{tab:layout}), which is what makes larger networks
representable on a single device at all.

\subsection*{Comparison with trained surrogate models}

The trained GAT approach depends structurally on the reference solver it is
designed to replace.
The training pipeline runs the reference solver at least 1{,}200 times to
generate labelled flows, trains the neural network on those flows, and then
deploys the network for inference.
The break-even threshold---the scenario count at which the full GAT pipeline
cost falls below that of running the reference solver directly---is
approximately 2{,}890 scenarios on Sioux Falls and 2{,}133 on
Barcelona~\cite{HuXie2025}.
The regime where GAT provides a net time saving is narrow and requires a
fixed, calibrated network evaluated many times with new OD inputs that stay
within the training distribution.
Network changes, post-construction updates, or parameter perturbations beyond
the training range all require retraining at the same upfront cost.
DiLLSUE has no training dependency.
The same algorithm, with the same hyperparameters ($L_{\mathrm{val}} = 50$,
tolerance $10^{-7}$, identical Newton settings), produced every result in
this paper without any network-specific modification. Only the dimensionless
sharpness $\tau$ was varied by design.

Beyond training cost, GAT-L predicts link flows independently without
structural coupling between inflow and outflow at any node.
The predicted flow state can violate the fundamental continuity constraint at
any intermediate node.
Hu and Xie acknowledge this explicitly in their Table~9 and
Section~5.1~\cite{HuXie2025}.
DiLLSUE satisfies flow conservation at every non-destination node by the
structure of the forward absorption step~\eqref{eq:absorption}, and every
solution it returns carries a certified equilibrium gap~\eqref{eq:gap}.
A surrogate's prediction carries no such certificate.
Trained surrogates are also specific to the network on which they were fitted.
No cross-network generalisation is reported in Hu and Xie~\cite{HuXie2025}.
DiLLSUE applies to any network from the network description alone.

\subsection*{Logit SUE as a behavioural model}

Wardrop UE assumes that all travellers have perfect information and identical
route preferences~\cite{Wardrop1952}.
This assumption has been questioned on behavioural grounds since Daganzo and
Sheffi formalised the stochastic alternative~\cite{DaganzoSheffi1977,Sheffi1985}.
Logit SUE replaces the perfect-information assumption with a probabilistic
route choice model in which travellers have heterogeneous and imperfect
perceptions of travel cost.
The logit SUE fixed point has been connected to maximum-entropy user equilibrium
under specific day-to-day behavioural update rules~\cite{LiEtAl2024}.
The practical dominance of Wardrop UE in planning practice has reflected the
absence of scalable SUE solvers, not a consensus that perfect information is
the better behavioural assumption.
DiLLSUE-FullGraph provides a scalable implementation that makes this
behavioural choice computationally accessible.

The equilibrium gap between logit SUE and Wardrop UE is characterised across
five networks and the full $\tau$ grid in Section~\ref{sec:bridge} and
Supplementary Section~S6.
At fixed $\tau$ the gap varies by orders of magnitude across networks and
cannot be predicted from volume-to-capacity ratio alone. On lightly loaded
networks it can exceed 100\% median MAPE at moderate sharpness, establishing
equilibrium model choice as a first-order modelling decision on most
benchmark networks.

\subsection*{Acyclicity filter landscape and alternatives}

The recursive logit formulation of Fosgerau et
al.~\cite{FosgerauEtAl2013} has been used almost exclusively for route choice
estimation: recovering utility parameters from observed individual route or
GPS path data~\cite{TranEtAl2025}.
In that context, observed routes are inherently acyclic, the active link set
is constrained by the data, and the acyclicity filter problem does not arise.
Forward assignment on a general cyclic urban network without observed paths
requires the Bellman to define its own active link set from topology alone.
This is the setting of the present paper, and no prior work constructs or
validates an acyclicity filter for this purpose.

Three acyclicity filters for link-based logit assignment can be distinguished
by the strictness of their forward-leaning criterion.
Dial's eligible-link filter~\cite{Dial1971} is the most restrictive, requiring
both origin-forward and destination-forward conditions per OD pair.
This makes the filter OD-specific and unsuitable as a global Bellman active
set.
The DSP filter (this paper) requires only the destination-forward condition,
strictly expanding the eligible set relative to Dial.
DSP is the maximal filter under the free-flow distance ordering: any excluded
link satisfies the reverse inequality, and adding it would create a directed
cycle, violating the DAG property.
The BFS filter uses hop count rather than travel time and is not a subset of
DSP: some links included by BFS are excluded by DSP and vice versa.
Section~S2 establishes that BFS produces a DAG but does not
guarantee connectivity.

A fundamentally different alternative is the cyclic Markov chain formulation
of Akamatsu~\cite{Akamatsu1996}, which eliminates the acyclicity filter
entirely.
Travellers choose the next link at each node according to logit probabilities,
generating an absorbing Markov chain with the destination as the absorbing
state.
Link flows are the expected traversal counts under this chain, obtained by
solving a linear system $(I - \mathbf{P}^d)^{-1}\mathbf{q}^d$ per destination.
This is well-posed on any strongly connected graph without a filter, and the
destination-batching architecture is preserved because the linear system still
operates on an $(|\mathcal{N}| \times |D|)$ tensor.
An extension to the Network GEV model preserves the same
structure~\cite{OyamaEtAl2022}.

The iterative Bellman backward pass and forward absorption of DiLLSUE-FullGraph
are the power-iteration implementation of the Akamatsu absorbing Markov chain.
Both converge to the same recursive logit equilibrium that Fosgerau
et al.~\cite{FosgerauEtAl2013} formalised with modern random utility
foundations, whose existence and uniqueness follow from Baillon and
Cominetti~\cite{BaillonCominetti2008}.
Prior computational treatments of this family solve small to medium networks
on CPU~\cite{Akamatsu1996,BaillonCominetti2008,OyamaEtAl2022}. A
destination-batched GPU implementation of the recursive logit assignment
itself, validated against an independent reference and the Wardrop limit on
standard congested benchmarks, is the gap this paper fills.

The Akamatsu model allows travellers to revisit nodes, so it is not a
numerically exact computation of the Fisk~\cite{Fisk1980} entropy-regularised
equilibrium over simple paths.
The DSP and BFS filters restrict the model to acyclic paths, recovering a
closer approximation to the Fisk simple-path logit SUE at the cost of the
filter approximation quantified in Supplementary Section~S7.

\subsection*{Limitations}
\label{sec:limitations}

The cyclic full-graph model is not well-posed everywhere.
Its fixed point exists only when the spectral radius of the
exponentiated-utility weight matrix is below one, which fails on Eastern
Massachusetts at $\tau \leq 3$. On networks dense in very short links the
admitted cyclic walks can traverse links many times, so that converged
full-graph flows on Winnipeg substantially exceed path-interpretable levels
at low $\tau$.
On Anaheim the Newton outer stalled above tolerance on all cells, which we
attribute to conditioning of that network's centroid-connector structure.
Resolving it is left to future work.
For all three situations the acyclic DSP variant is the practical answer and
is reported throughout. This is the same restriction-for-well-posedness
trade-off that motivates prism-constrained formulations~\cite{OyamaEtAl2022}.
The eDSP refresh, effective on Eastern Massachusetts, can oscillate between
successive DAGs on symmetric congested networks (it converged on 68--99\% of
Sioux Falls scenarios depending on $\tau$). Its stability envelope is
characterised in Supplementary Section~S4.

DiLLSUE's robustness under partial or missing OD data is not characterised in
this study.
Hu and Xie~\cite{HuXie2025} demonstrate that GAT maintains reasonable accuracy
under 5--20\% missing OD data, which represents a genuine advantage of the
trained surrogate in data-constrained settings.

Barcelona and Winnipeg use uniform unit link capacities in the TNTP source
files. Under the perturbed demand both operate in the saturated regime of the
capped BPR function, where costs decouple from flows over much of the
network.
The large SUE--UE gaps reported for these networks at all $\tau$ are a
property of that regime rather than of the solver.

Wall-clock results are reported for per-scenario Newton solves on networks
one to two orders of magnitude smaller than planning-scale models. Scenario
batching and a planning-scale demonstration are left to future work.

\subsection*{Future directions}

Four directions follow directly from the findings.
First, empirical calibration of $\mu$ from observed traffic counts would
establish the range of practically relevant dispersion values and connect the
logit SUE model to measured network behaviour.
Second, scaling to planning-size networks: the per-destination tensor of a
metropolitan network exceeds single-device memory under Newton's unrolled
inner passes, and the dual cost-space formulations that decouple the
optimisation variable from the destination count~\cite{OyamaEtAl2022},
combined with the Anderson outer that our ablation found
fastest per iteration, are the natural route. The destination-batched loading
map developed here remains the inner engine in that design.
Third, the eDSP outer iteration provides a validated mechanism for closing
the filter approximation gap. Establishing its convergence conditions
analytically remains an open theoretical problem.
Fourth, recovering BPR parameters or the dispersion parameter $\mu$ from
observed link flows via gradient descent through the solver's computation graph
is the natural parameter estimation extension, enabled by the differentiable
formulation~\cite{NieLi2026}.

\section{Conclusion}
\label{sec:conclusion}

DiLLSUE is a zero-training solver for the link-based recursive logit SUE and,
to our knowledge, its first GPU implementation.
It combines destination-batched recursive logit Bellman loading with a
merit-safeguarded IFT-Newton outer solver selected through a systematic
four-way evaluation of outer solvers on the identical loading map.
The link-level formulation eliminates route enumeration and exposes a
fixed-shape $(|\mathcal{N}| \times |D|)$ value tensor that GPU operations
exploit directly.
The DSP acyclicity filter, proved to produce a connected DAG for every
destination (Supplementary Material, Section~S2), is evaluated as an
optional single-pass approximation, and its equilibrium-cost refresh (eDSP)
closes most of its approximation gap where the refresh is stable.

Three empirical findings on five standard open benchmark networks support the
solver's practical use.
First, GPU and CPU executions of the same algorithm agree at the
floating-point floor (worst-cell mean MAPE $6\times10^{-5}$\,\%,
$R^2 = 1.000$), confirming that the destination-batched implementation
computes the same recursive logit SUE fixed point on both platforms.
Second, the computed SUE converges monotonically to the independently
computed Wardrop UE as the dimensionless sharpness $\tau$ grows, reaching
0.33\% MAPE on Sioux Falls at $\tau = 100$ --- the model-correctness anchor
of the study, and a cell in which the zero-training solver is more accurate
than the published trained GAT surrogates of Hu and Xie~\cite{HuXie2025} on
identical inputs.
Third, the GPU advantage of destination batching grows with network size,
as the structural argument predicts, while every solution carries a certified
equilibrium gap that no surrogate provides.

The differentiable formulation makes DiLLSUE-FullGraph directly extensible
to parameter estimation and sensitivity analysis via automatic differentiation,
opening a route from forward assignment to data-driven model calibration
without changing the core solver.

\section*{Acknowledgements}

This research was supported by the Cambridge Commonwealth, European and
International Trust.
Additional support was provided by the Martin Centre for Architectural and
Urban Studies, University of Cambridge.
The authors thank Ben Stabler, Hillel Bar-Gera, and Elizabeth Sall for
maintaining the TransportationNetworks benchmark
repository~\cite{Stabler2023}.

\section*{Declaration of competing interest}

The authors declare that they have no known competing financial interests or
personal relationships that could have appeared to influence the work reported
in this paper.

\section*{Data availability}

All network data used in this study are publicly available from the
TransportationNetworks repository~\cite{Stabler2023}
(\url{https://github.com/bstabler/TransportationNetworks}).
The repository contains the network topology, demand matrices, and link
attribute files for all five benchmark networks used in this study.
No proprietary or restricted data were used.

\section*{Code availability}

All code required to reproduce the experiments reported in this paper,
including the DiLLSUE GPU solver and its filter variants, the CPU execution
path, the Ref-UE-CFW Frank-Wolfe solver, scenario generation, and all
plotting scripts, is available at
\url{https://github.com/yueli901/Mukara5}.


\clearpage
\section*{Supplementary Material}
\setcounter{section}{0}
\setcounter{table}{0}
\setcounter{figure}{0}
\setcounter{equation}{0}
\setcounter{proposition}{0}
\setcounter{corollary}{0}
\setcounter{theorem}{0}
\renewcommand\thesection{S\arabic{section}}
\renewcommand\thetable{S\arabic{table}}
\renewcommand\thefigure{S\arabic{figure}}
\renewcommand\theequation{S\arabic{equation}}
\renewcommand\theproposition{S\arabic{section}.\arabic{proposition}}
\renewcommand\thetheorem{S\arabic{section}.\arabic{theorem}}

\section{Self-regulating averaging: convergence analysis}
\label{sec:S1}

The method of successive averages (MSA) with the harmonic step
$\alpha_t = 1/t$~\cite{SheffiPowell1982,Sheffi1985} is the textbook
algorithm for link-based SUE.
Its convergence to the fixed point
$\mathbf{x}^* = \mathbf{f}(c(\mathbf{x}^*))$ is guaranteed because
$\sum_t 1/t = \infty$ and $\sum_t 1/t^2 < \infty$~\cite{RobbinsMonro1951}.
In practice, the convergence rate on the L2 fixed-point relative residual
$\mathrm{gap}_{\mathrm{rel}}(\mathbf{x}) =
\|\mathbf{x} - \mathbf{f}(c(\mathbf{x}))\|_2 / \|\mathbf{x}\|_2$
is far too slow for a production threshold of $10^{-7}$.

Table~\ref{tab:S1_msa_stall} reports MSA convergence on Sioux Falls
($\mu = 1.0$, single scenario).

\begin{table}[h]
\centering
\caption{MSA convergence on Sioux Falls ($\mu = 1.0$, single scenario,
  76 links, 24 zones). The solver is capped at 10{,}000 iterations and
  does not reach $\mathrm{gap}_{\mathrm{rel}} < 10^{-5}$, let alone the
  production threshold of $10^{-7}$. Values measured under the $L_\infty$
  criterion (pre-production); $\mathrm{gap}_{\mathrm{rel}}$ follows a
  qualitatively identical trend.}
\label{tab:S1_msa_stall}
\begin{tabular}{rr}
\toprule
Iteration & $\mathrm{gap}_{\mathrm{rel}}$ (L$_\infty$, pre-prod.) \\
\midrule
    10 & $4.62 \times 10^{-1}$ \\
   100 & $3.56 \times 10^{-2}$ \\
 1{,}000 & $3.72 \times 10^{-3}$ \\
 5{,}000 & $7.52 \times 10^{-4}$ \\
10{,}000 & $3.77 \times 10^{-4}$ (cap; not converged) \\
\bottomrule
\end{tabular}
\end{table}

Reaching $\mathrm{gap}_{\mathrm{rel}} < 10^{-5}$ under plain MSA would require an
estimated $10^5$ iterations, which is computationally infeasible at planning
scale; the production threshold of $10^{-7}$ is further out of reach.

The SRA algorithm~\cite{LiuEtAl2009} replaces the harmonic step with a
backtracking line search on the squared fixed-point residual merit function
$G(\mathbf{x}) = \|\mathbf{x} - \mathbf{f}(c(\mathbf{x}))\|_2^2$.
The update direction is $\mathbf{d}_t = \mathbf{f}(c(\mathbf{x}_t)) -
\mathbf{x}_t$ (the Frank-Wolfe direction, produced as a by-product of the
Bellman and absorption passes).
The step $\alpha_t$ is selected by halving from 1 until
$G(\mathbf{x}_t + \alpha_t \mathbf{d}_t) < G(\mathbf{x}_t)$ holds, with
MSA fallback $\alpha_t = 1/t$ if no halving succeeds within eight attempts.
The fallback preserves the Sheffi--Powell convergence
guarantee~\cite{SheffiPowell1982,Cantarella1997}.

Beyond SRA, two further outer solvers were evaluated in a production
four-way comparison on Sioux Falls at $\tau \in \{3, 10\}$ (fp64, 200
scenarios per cell, GPU): type-II Anderson mixing with window
$m = 5$~\cite{Walker2011} and the production IFT-Newton of the main paper,
all four sharing the identical dense Bellman and forward-absorption inner
passes.
Figure~\ref{fig:S_ladder} and Table~\ref{tab:S1_ladder} report the results.

\begin{figure}[h]
\centering
\includegraphics[width=0.6\textwidth]{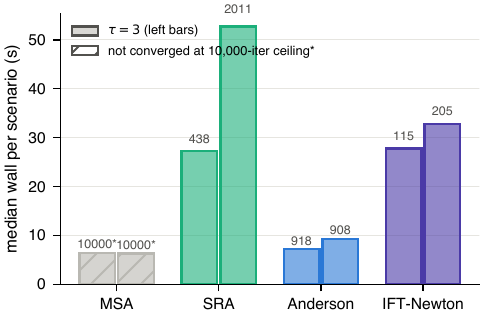}
\caption{Outer-solver comparison on Sioux Falls ($\tau \in \{3, 10\}$,
  fp64, 200 scenarios per cell). Bars: median wall time per scenario to
  $\mathrm{gap}_{\mathrm{rel}} < 10^{-7}$; annotations: median iteration
  counts. MSA bars (hatched) show wall time at the 10{,}000-iteration
  ceiling without converging. Newton iteration counts are outer Newton
  steps, each comprising a converged loading evaluation plus a GMRES solve,
  and are not directly comparable to the fixed-point iteration counts of
  the other three solvers.}
\label{fig:S_ladder}
\end{figure}

\begin{table}[h]
\centering
\caption{Outer-solver ladder on Sioux Falls (fp64, 200 scenarios per cell,
  median over scenarios). Converged = fraction of scenarios reaching
  $\mathrm{gap}_{\mathrm{rel}} < 10^{-7}$ within the 10{,}000-iteration
  ceiling.}
\label{tab:S1_ladder}
\begin{tabular}{llrrr}
\toprule
Outer solver & $\tau$ & Median iters & Median wall (s) & Converged \\
\midrule
MSA        & 3  & 10{,}000 & 6.4  & 0\%   \\
MSA        & 10 & 10{,}000 & 6.3  & 0\%   \\
SRA        & 3  & 438      & 27.2 & 100\% \\
SRA        & 10 & 2{,}011  & 52.8 & 98\%  \\
Anderson   & 3  & 918      & 7.1  & 100\% \\
Anderson   & 10 & 908      & 9.2  & 100\% \\
IFT-Newton & 3  & 115      & 27.8 & 100\% \\
IFT-Newton & 10 & 205      & 32.7 & 100\% \\
\bottomrule
\end{tabular}
\end{table}

Three observations drive the production choice.
Plain MSA is cheapest per iteration but never reaches the production
threshold.
Anderson mixing is the fastest to tolerance in wall time on this network:
its extrapolation approximates the Newton step from the residual history ---
a matrix-free quasi-Newton --- at the cost of a small $O(m^2)$ least-squares
solve per iteration.
IFT-Newton uses one to two orders of magnitude fewer outer steps than the
extrapolation methods; each step is expensive (a converged inner loading
plus a GMRES solve), so its wall time on this easy cell is higher than
Anderson's.
Newton's advantage is uniformity: its merit-safeguarded steps converge on
every reported cell of the full $\tau$ grid, including the
near-deterministic regime in which the fixed-point map loses contraction
and extrapolation-based accelerators degrade.
The production runs of the main paper therefore use IFT-Newton throughout,
accepting a wall-time premium on easy cells for uniform convergence over
the grid.
A convergence certificate for a safeguarded variant of the Anderson
iteration is given in Section~\ref{sec:S8}.

\section{Acyclicity filter: formal properties}
\label{sec:S2}

The destination shortest-path (DSP) filter defines the active link set for
destination $d$ from base free-flow costs:
\begin{equation}
\mathcal{E}^d_{\mathrm{DSP}} =
  \bigl\{(u,v)\in\mathcal{E} : d_{\mathrm{FF}}(v,d) < d_{\mathrm{FF}}(u,d)\bigr\},
\label{eq:S_dsp}
\end{equation}
where $d_{\mathrm{FF}}(n,d)$ is the free-flow shortest-path distance from
node $n$ to destination $d$, computed by Dijkstra on the reversed network.
The filter is computed once per network and never updated during assignment.

\begin{proposition}[DSP DAG and connectivity]
\label{prop:S2_dsp_dag}
For any destination $d\in D$, the DSP-filtered subgraph
$(\mathcal{N},\mathcal{E}^d_{\mathrm{DSP}})$ is a directed acyclic graph.
Moreover, for every origin $o$ reachable from $d$ in free-flow, at least one
$o\to d$ path exists in $\mathcal{E}^d_{\mathrm{DSP}}$.
\end{proposition}

\begin{proof}
\emph{Acyclicity.}
Suppose $v_0 \to v_1 \to \cdots \to v_k \to v_0$ is a directed cycle in
$\mathcal{E}^d_{\mathrm{DSP}}$.
Each link satisfies $d_{\mathrm{FF}}(v_{m+1},d) < d_{\mathrm{FF}}(v_m,d)$
(indices mod $k{+}1$).
Summing around the cycle gives
$d_{\mathrm{FF}}(v_0,d) < d_{\mathrm{FF}}(v_0,d)$, a contradiction.

\emph{Connectivity.}
Let $\pi = (o = w_0, w_1, \ldots, w_K = d)$ be the free-flow shortest path.
For every consecutive pair $(w_i,w_{i+1})$ on $\pi$, the triangle-inequality
optimality condition gives
$d_{\mathrm{FF}}(w_i,d) = c_{(w_i,w_{i+1})} + d_{\mathrm{FF}}(w_{i+1},d)$.
Since $c_{(w_i,w_{i+1})} > 0$, it follows that
$d_{\mathrm{FF}}(w_{i+1},d) < d_{\mathrm{FF}}(w_i,d)$, so
$(w_i,w_{i+1}) \in \mathcal{E}^d_{\mathrm{DSP}}$.
Hence $\pi$ is a feasible $o\to d$ path in the DSP-filtered subgraph.
\end{proof}

Because $\mathcal{E}^d_{\mathrm{DSP}}$ is a DAG, the Bellman backward sweep
converges in a single topological pass.
Logit SUE is unique on any directed acyclic subgraph~\cite{Fisk1980}, so the
DSP-filtered equilibrium is well-defined.

The BFS filter uses hop count in place of travel-time distances:
\begin{equation}
\mathcal{E}^d_{\mathrm{BFS}} =
  \bigl\{(u,v)\in\mathcal{E} : h_d(v) < h_d(u)\bigr\},
\label{eq:S_bfs}
\end{equation}
where $h_d(n)$ is the minimum hop count from $n$ to $d$ on the reversed
network, computed by BFS.
The BFS filter is topology-invariant and requires only a single BFS per destination.

\begin{corollary}[BFS DAG, connectivity uncertain]
\label{cor:S2_bfs_dag}
The BFS-filtered subgraph $(\mathcal{N},\mathcal{E}^d_{\mathrm{BFS}})$
is a directed acyclic graph for every destination $d$.
Connectivity is not guaranteed: a link $(u,v)$ with $h_d(u) = h_d(v)$ is
excluded even when it lies on the only available $u\to d$ route.
\end{corollary}

\begin{proof}
Acyclicity follows by the same contradiction argument as above applied to
integer hop counts.
For the connectivity failure: a link $(u,v)$ satisfying $h_d(u) = h_d(v)$
is excluded from $\mathcal{E}^d_{\mathrm{BFS}}$ even when it lies on the
only available route to $d$.
The DSP filter avoids this because every link on the free-flow shortest path
satisfies $d_{\mathrm{FF}}(\mathrm{head},d) < d_{\mathrm{FF}}(\mathrm{tail},d)$
strictly (Proposition~\ref{prop:S2_dsp_dag}).
\end{proof}

Empirical convergence of the BFS filter on all five benchmark networks is
documented in Section~\ref{sec:S3}.

\section{BFS acyclicity filter: results}
\label{sec:S3}

The BFS filter defines a topology-based acyclic subgraph that does not require
free-flow shortest-path distances.
It is computationally cheaper to construct than DSP (BFS on reversed network
versus Dijkstra), but the connectivity failure mode established in
Section~\ref{sec:S2} is a theoretical concern.

The hop-count tie failure mode does not arise on any of the five benchmark
networks tested here.
On Barcelona and Winnipeg, where the TNTP source files assign uniform unit
capacities, the BPR cost function with large $x/c$ ratios requires the
clamp $x/c \leq 100$ applied to all link flows before the power operation;
after this correction, both BFS and DSP reach convergence tolerance on all
scenarios.

Table~\ref{tab:S3_bfs} reports the production BFS run at $\tau = 10$
(fp64, tol $10^{-7}$, 200 scenarios per network) against the CFW Wardrop
UE reference, alongside the DSP variant at the same $\tau$.
BFS converges on every scenario of every network, but its hop-count
criterion is a substantially coarser proxy for cost-efficiency than the
free-flow shortest-path criterion: its distance to the UE reference is
1.7 times DSP's on Sioux Falls and 2 to 4 times DSP's on the larger
networks.
BFS is therefore retained only as a construction-cost baseline; DSP
dominates it at equal per-iteration cost.

\begin{table}[h]
\centering
\caption{BFS versus DSP at $\tau = 10$: median MAPE\,(\%) of converged link
  flows against the CFW Wardrop UE reference, 200 scenarios per network.}
\label{tab:S3_bfs}
\begin{tabular}{lrr}
\toprule
Network & BFS & DSP \\
\midrule
Sioux Falls           & 39.3  & 23.9 \\
Eastern Massachusetts & 84.4  & 33.1 \\
Anaheim               & 250.2 & 80.3 \\
Barcelona             & 386.4 & 255.5 \\
Winnipeg              & 468.3 & 148.6 \\
\bottomrule
\end{tabular}
\end{table}

\section{Equilibrium-cost DSP (eDSP): outer-loop refinement}
\label{sec:S4}

The DSP filter is constructed from free-flow travel times.
At congested equilibrium, some links that increase free-flow distance to a
destination may become shorter in equilibrium travel time.
Replacing free-flow with equilibrium-cost distances in the DSP criterion
defines the equilibrium-cost DSP (eDSP), which includes those
backward-in-free-flow but forward-in-equilibrium links.

The outer-loop algorithm is:
\begin{align*}
M^{(0)}    &= \mathrm{DSP}(t^0), \\
\mathbf{x}^{(k)} &= \text{DiLLSUE}\!\left(M^{(k)}\right), \\
c^{(k)}    &= \mathrm{BPR}\!\left(\mathbf{x}^{(k)}\right), \\
M^{(k+1)}  &= \mathrm{DSP}\!\left(c^{(k)}\right),
\end{align*}
stopping when $M^{(k+1)} = M^{(k)}$ (byte-identical active-link set)
or after eight outer iterations.
Each inner DiLLSUE solve uses $L_{\mathrm{dtd}} = 3{,}000$ and
$L_{\mathrm{val}} = 50$.

The path-enumeration logit SUE from Section~\ref{sec:S7} serves as ground
truth for the results below.
Table~\ref{tab:S4_edsp} reports per-outer-iteration MAPE against this
baseline on Sioux Falls.

\begin{table}[h]
\centering
\caption{Equilibrium-cost DSP outer-loop convergence on Sioux Falls.
  Each cell is MAPE\,(\%) of DiLLSUE link flow $\mathbf{x}^{(k)}$ against
  the path-enumeration logit SUE baseline.
  Row $k=0$ uses the free-flow DSP filter (reproduces
  Table~\ref{tab:S7_filter_approx}).
  Dashes denote mask convergence.}
\label{tab:S4_edsp}
\begin{tabular}{rrrrrr}
\toprule
Outer $k$ & $\mu=0.1$ & $\mu=0.5$ & $\mu=1.0$ & $\mu=2.0$ & $\mu=5.0$ \\
\midrule
0 & 20.73 & 10.62 & 9.73  & 9.21  & 8.82  \\
1 & 15.63 & 3.37  & 3.60  & 3.93  & 4.10  \\
2 & 19.60 & 2.22  & 2.45  & 2.36  & 2.31  \\
3 & 16.41 & 1.79  & 0.10  & 0.024 & 0.022 \\
4 & 23.71 & 1.43  & 0.10  & 0.025 & 0.020 \\
5 & 19.82 & 0.74  & 0.10  & --    & --    \\
6 & 26.79 & 0.75  & --    & --    & --    \\
7 & 24.32 & 0.75  & --    & --    & --    \\
\bottomrule
\end{tabular}
\end{table}

Three convergence regimes emerge by $\mu$.
At high dispersion ($\mu \geq 1.0$), the outer loop converges in three
iterations to MAPE below 0.1\% and the mask stabilises thereafter.
At intermediate dispersion ($\mu = 0.5$), convergence is monotone but slow,
reaching 0.75\% MAPE by $k=6$.
At low dispersion ($\mu = 0.1$), the outer loop does not converge: MAPE
oscillates and the mask-difference count grows rather than shrinking.
The free-flow DSP filter used throughout the main paper remains the practical
default for single-pass solves.

In the production $\tau$-grid runs, the refresh's value and its stability
envelope both appear.
On Eastern Massachusetts, whose asymmetric topology leaves the free-flow
DAG far from the congested one, eDSP closes the DSP gap against the
converged FullGraph equilibrium from 21.5\% median MAPE at $\tau = 10$
(DSP level) to 4.7\% at $\tau = 30$ and 2.9\% at $\tau = 100$.
On the symmetric, heavily congested Sioux Falls network the rebuilt DAGs
can alternate between near-tied configurations: the refresh converged on
68--99\% of scenarios depending on $\tau$ (135 to 199 of 200), the
non-converged scenarios oscillating between masks rather than diverging.
A damped or hysteretic mask update is the natural remedy and is left to
future work.

\section{Convergence threshold and numerical precision}
\label{sec:S5}

Two implementation choices are linked and jointly diagnosed here: the
convergence threshold $\mathrm{tol} = 10^{-7}$ (tighter than the engineering
default $10^{-5}$) and double-precision arithmetic (fp64) rather than the
PyTorch default fp32.
Neither choice is independent of the other, and both are driven by the
flat-BPR characteristic of Barcelona and Winnipeg.

\paragraph*{Why tol~$= 10^{-7}$}

Barcelona and Winnipeg have BPR coefficients $b \approx 10^{-9}$ in TNTP
normalised units; the BPR cost factor $1 + a(x/c)^b$ remains close to 1 for
almost all link loadings.
The loading map $\mathbf{f}$ is therefore near-identity on most of the
network, and the basin of attraction of the SUE fixed point is wide.

At $\mathrm{tol} = 10^{-5}$ on Barcelona ($\mu = 1.0$, single diagnostic
cell):
\begin{itemize}
\item DiLLSUE-FullGraph (GPU) stopped at outer iteration $p_{50} = 9$,
  $\mathrm{gap}_{p_{50}} \approx 4 \times 10^{-6}$.
\item Ref-SUE-FullGraph (CPU) stopped at outer iteration $p_{50} = 3$,
  $\mathrm{gap}_{p_{50}} \approx 9 \times 10^{-6}$.
\item Both passed the $< 10^{-5}$ test; both reported ``converged.''
\item Total flow agreed to 5 significant figures on both sides.
\item \textbf{Per-link WMAPE: 113.9\%} between the two implementations.
\end{itemize}

The two solvers landed on different epsilon-fixed-points within the same flat
basin.
Tightening to $\mathrm{tol} = 10^{-7}$ forces both past the flat-basin
region:
\begin{itemize}
\item Both reached $p_{50} = 12$ outer iterations, $\mathrm{max} = 136$
  (identical distribution).
\item Total flow agreed to 7 significant figures.
\item \textbf{Per-link WMAPE: 0.000008\%} (fp64 round-off floor).
\end{itemize}

Tightening by two orders of magnitude in the criterion reduced CPU/GPU
disagreement by five orders of magnitude.
The compute cost is modest: Barcelona moved from 3 (basin-mismatch artifact)
to 12 outer iterations, but per-iteration time dropped $9\times$ from the
per-destination layout switch, so wall-clock at $10^{-7}$ is faster than the
old code at $10^{-5}$.
The criterion formula is unchanged (Sheffi 1985, Cantarella 1997); only the
threshold is tighter.

\paragraph*{Why fp64}

Four considerations jointly motivate fp64 throughout.

\emph{CPU/GPU comparability.}
The CPU reference solver (Ref-SUE-FullGraph) operates in
\texttt{numpy.float64}.
Running fp32 on the GPU would inject a dtype-induced precision floor of
roughly 0.001--0.01\% WMAPE into the Comparison~1 correctness check, with no
methodological justification.

\emph{tol~$= 10^{-7}$ requires fp64.}
fp32 carries roughly seven decimal digits of precision; its residual noise
floor is $\sim\!10^{-7}$ at best.
Running fp32 against a $10^{-7}$ stopping criterion places the threshold at
the precision floor: many scenarios would hover stochastically near the
criterion and never converge within the $10{,}000$-iteration ceiling.
fp32 and $\mathrm{tol} = 10^{-7}$ are not jointly feasible.

\emph{Barcelona BPR overflow safety.}
Barcelona's per-link BPR coefficient $b$ reaches 16.83 with capacities near
unity.
With the ratio clamp at 100 (applied before the power operation), the
worst-case BPR cost factor is approximately $100^{16.83} \approx 4.6 \times
10^{33}$.
fp32 represents up to $\sim\!3.4 \times 10^{38}$, leaving only four orders
of headroom.
A prior incident (Mukara4, May 2026) documented fp32 overflow on Barcelona
during autograd, invalidating 324 cells silently; fp64 removes this
fragility.

\emph{The headline speedup does not require fp32.}
After the per-destination layout switch, DiLLSUE-FullGraph achieves the
reported GPU speedup at fp64.
Switching to fp32 would nominally improve throughput but at the cost of
CPU/GPU comparability and convergence properties above, without changing any
qualitative finding.

\section{SUE versus UE equilibrium gap}
\label{sec:S6}

The comparisons in this section quantify the distance between two equilibrium
concepts: logit stochastic user equilibrium (SUE) under the recursive logit
formulation and Wardrop user equilibrium (UE) as computed by the Frank-Wolfe
method.
They measure a property of the models, not solver error.

Table~\ref{tab:S6_tau_sensitivity} reports the equilibrium gap as a
function of the dimensionless sharpness $\tau = \mu\bar{c}$ across all five
networks (main-paper equation for $\tau$; per-network $\mu = \tau/\bar{c}$
from the main paper's network table).
As $\tau$ increases, logit probabilities concentrate on lower-cost routes
and the gap decreases monotonically, consistent with the theoretical
convergence of logit SUE to Wardrop UE as $\mu \to \infty$.

\begin{table}[h]
\centering
\caption{SUE/UE equilibrium gap (median MAPE,\,\%, converged scenarios of
  200 per cell) as a function of $\tau$.
  FullGraph within its well-posed scope; DSP on all five networks.
  This is a model comparison, not a solver accuracy measure: on the
  saturated unit-capacity networks (Barcelona, Winnipeg) the capped BPR
  cost decouples costs from flows over much of the network and the concept
  gap remains large at every $\tau$.}
\label{tab:S6_tau_sensitivity}
\begin{tabular}{llrrrrr}
\toprule
Network & Variant & $\tau=1$ & $\tau=3$ & $\tau=10$ & $\tau=30$ & $\tau=100$ \\
\midrule
Sioux Falls           & FullGraph & 26.9 & 6.0 & 1.0 & 0.36 & 0.12 \\
Sioux Falls           & DSP       & 24.5 & 24.1 & 23.9 & 23.9 & 23.8 \\
Eastern Massachusetts & FullGraph & --   & --  & 57.6 & 14.0 & 4.7 \\
Eastern Massachusetts & DSP       & 73.1 & 56.2 & 33.1 & 21.2 & 15.1 \\
Anaheim               & DSP       & 107.7 & 98.6 & 80.3 & 62.5 & 51.3 \\
Barcelona             & DSP       & 344.9 & 321.6 & 255.5 & 153.4 & 59.5 \\
Winnipeg              & DSP       & 179.5 & 171.9 & 148.6 & 108.9 & 71.1 \\
\bottomrule
\end{tabular}
\end{table}

Two patterns stand out.
On congested networks with informative capacities (Sioux Falls, Eastern
Massachusetts), the exact FullGraph gap collapses toward zero as $\tau$
grows, while the DSP gap flattens at the filter's approximation ceiling.
On lightly loaded or saturated networks the gap at fixed $\tau$ is far
larger and closes slowly: route costs are nearly flow-independent there, so
logit choice spreads flow across many near-tied routes that deterministic
assignment never uses.
The SUE/UE gap therefore cannot be predicted from volume-to-capacity ratio
alone, and equilibrium model choice is a first-order modelling decision on
most benchmark networks.

\section{DSP filter: path coverage and approximation gap}
\label{sec:S7}

The DSP filter restricts the active path set to routes whose every link
strictly reduces free-flow distance to the destination.
Paths using backward links receive zero probability under DSP recursive logit.
In full path-enumeration logit SUE, these paths receive positive probability
proportional to $\exp(-\mu \cdot \mathrm{cost})$.
The DSP filter therefore introduces a modelling approximation whose magnitude
decreases with $\mu$ and vanishes as $\mu \to \infty$.

Note: path-enumeration SUE and recursive logit SUE are two distinct models
even without the DSP filter.
Path enumeration requires acyclic paths by definition; recursive logit
operates on the full graph via the absorbing Markov chain.
The comparison in this section is model-to-model at the level of the DSP
approximation, not a solver accuracy comparison.

On Sioux Falls (24 nodes, 76 links, 24 zones), all simple paths between each
OD pair were enumerated by depth-first search with a per-OD cap of 5{,}000
paths.
The cap was never reached: the network admits 1{,}632{,}820 distinct simple
paths across 528 active OD pairs (mean: 3{,}092 paths per OD pair).
A path-based logit SUE was solved by 5{,}000-iteration path-MSA with relative
convergence tolerance $10^{-5}$ at each
$\mu \in \{0.1,\,0.5,\,1.0,\,2.0,\,5.0\}$, using base OD demand and native
TNTP BPR parameters ($a = 0.15$, $b = 4$ on all links).
DiLLSUE-DSP was solved at the same $\mu$ values with $L_{\mathrm{dtd}} =
3{,}000$.

Table~\ref{tab:S7_filter_approx} reports the MAPE of DSP recursive logit
SUE against the path-enumeration baseline.

\begin{table}[h]
\centering
\caption{DSP filter approximation error on Sioux Falls: MAPE of
  DiLLSUE-DSP against full path-enumeration logit SUE
  (1{,}632{,}820 paths across 528 OD pairs), base OD demand,
  $L_{\mathrm{dtd}} = 3{,}000$, $L_{\mathrm{val}} = 50$.
  Total flow ratio is total DSP link flow divided by total
  path-enumeration link flow.}
\label{tab:S7_filter_approx}
\begin{tabular}{rrr}
\toprule
$\mu$ & MAPE\,(\%) & Total flow ratio \\
\midrule
0.1 & 20.73 & 0.779 \\
0.5 & 10.62 & 0.957 \\
1.0 &  9.73 & 0.973 \\
2.0 &  9.21 & 0.978 \\
5.0 &  8.82 & 0.979 \\
\bottomrule
\end{tabular}
\end{table}

The DSP recursive logit underestimates total flow by 22\% at $\mu = 0.1$
and by 2\% at $\mu \geq 1.0$.
Per-link MAPE decreases from 20.73\% at $\mu = 0.1$ to approximately 8.8\%
at $\mu = 5.0$ and does not approach zero, because the DSP filter
permanently excludes backward links regardless of their equilibrium
relevance.
For larger benchmark networks where full path enumeration is computationally
intractable, the relative magnitude of this approximation is unknown; the
$\approx\!10\%$ measured here is a lower bound.

\section{Convergence certificate for safeguarded Anderson mixing}
\label{sec:S8}

This section gives a convergence certificate for Anderson mixing applied to
the logit SUE fixed point, for a variant safeguarded by sufficient decrease
of the entropy-regularised SUE objective.
The production implementation uses a cheaper residual-based safeguard.
The certificate covers the Armijo-safeguarded variant exactly and the
production variant only empirically, and the gap is stated explicitly at the
end of the section.

\paragraph*{Setting}
Let $F$ denote the SUE loading map at dispersion $\mu > 0$ and
$\mathbf{r}(\mathbf{x}) = F(\mathbf{x}) - \mathbf{x}$ the fixed-point
residual.

\noindent\textbf{Assumption A1.}
Every link cost $c_e(\cdot)$ is continuously differentiable and strictly
increasing on $[0, \infty)$.

Under A1 the entropy-regularised objective of
Fisk~\cite{Fisk1980},
\begin{equation}
W_F(\mathbf{x}) = \sum_e \int_0^{x_e} c_e(v)\,dv
  \;-\; \frac{1}{\mu}\, H(\mathbf{x}),
\label{eq:S8_fisk}
\end{equation}
with $H$ the logit entropy at loading $\mathbf{x}$, is strictly convex and
continuously differentiable, and its unique minimiser $\mathbf{x}^*$ is the
SUE fixed point, $F(\mathbf{x}^*) = \mathbf{x}^*$~\cite{Fisk1980}.
The link-based counterpart of this objective for the cyclic Markovian model
is the entropy decomposition of Akamatsu~\cite{Akamatsu1996}.
Uniqueness of $\mathbf{x}^*$ is Fisk's result and is cited, not claimed.
The residual is a strict descent direction for $W_F$ away from the
equilibrium: $-\nabla W_F(\mathbf{x}) \cdot \mathbf{r}(\mathbf{x}) > 0$ for
$\mathbf{x} \neq \mathbf{x}^*$~\cite{Fisk1980,Sheffi1985}.

\paragraph*{Safeguarded algorithm}
Fix $\sigma \in (0,1)$ and $\sigma_{\mathrm{BT}} \in (0,1)$.
At each step $t$:
(1)~backtracking reference step: starting from $\alpha = 1$, halve $\alpha$
until
$W_F(\mathbf{x}_t + \alpha\,\mathbf{r}_t) \leq W_F(\mathbf{x}_t)
 - \sigma_{\mathrm{BT}}\,\alpha\,
 \bigl(-\nabla W_F(\mathbf{x}_t)\cdot\mathbf{r}_t\bigr)$,
giving $\mathbf{x}^{\mathrm{BT}}_{t+1}$ and decrease
$\Delta_t = W_F(\mathbf{x}_t) - W_F(\mathbf{x}^{\mathrm{BT}}_{t+1}) > 0$.
(2)~Anderson candidate $\mathbf{x}^{\mathrm{AA}}_{t+1}$ by the type-II
update of the main paper.
(3)~Armijo acceptance: take the candidate if
$W_F(\mathbf{x}_t) - W_F(\mathbf{x}^{\mathrm{AA}}_{t+1})
 \geq \sigma\,\Delta_t$, and the backtracking step otherwise.

\begin{theorem}[Safeguarded Anderson convergence]
\label{thm:S8_anderson}
Under A1, the safeguarded iteration converges to the unique SUE equilibrium
$\mathbf{x}^*$ from any feasible $\mathbf{x}_0 \geq \mathbf{0}$.
\end{theorem}

\begin{proof}
Every accepted step decreases $W_F$ by at least $\sigma \Delta_t$, and
$\Delta_t > 0$ whenever $\mathbf{x}_t \neq \mathbf{x}^*$ because
$\mathbf{r}_t$ is a strict descent direction there.
$W_F$ is bounded below by $W_F(\mathbf{x}^*)$, so
$W_F(\mathbf{x}_t) \to L \geq W_F(\mathbf{x}^*)$.
Suppose $L > W_F(\mathbf{x}^*)$.
Then the iterates remain outside some $\varepsilon$-ball of $\mathbf{x}^*$.
The feasible region is compact, the loading map keeps every active-link flow
strictly positive, and on the resulting compact interior set $W_F$ is
$L_W$-smooth with
$-\nabla W_F(\mathbf{x})\cdot\mathbf{r}(\mathbf{x}) \geq c_\varepsilon > 0$
by strict convexity and compactness.
Standard Armijo analysis then bounds the accepted backtracking step below,
$\alpha_t \geq 2(1-\sigma_{\mathrm{BT}})\,c_\varepsilon / (L_W M^2) > 0$
with $M = \sup \|\mathbf{r}\| < \infty$, so
$\Delta_t \geq \delta_\varepsilon > 0$ uniformly.
Summing the accepted decreases gives
$W_F(\mathbf{x}_T) \leq W_F(\mathbf{x}_0) - T \sigma \delta_\varepsilon
\to -\infty$, contradicting boundedness below.
Hence $L = W_F(\mathbf{x}^*)$, and strict convexity with compactness gives
$\mathbf{x}_t \to \mathbf{x}^*$.
\end{proof}

\begin{theorem}[Catastrophic Anderson failure ruled out]
\label{thm:S8_cycling}
Under A1, the iterate cannot cycle with
$\liminf_t \|\mathbf{r}_t\| > 0$.
The catastrophic Anderson failure mode, in which the history matrix stays
ill-conditioned while the residual does not converge, therefore cannot
arise.
\end{theorem}

\begin{proof}
If $\liminf_t \|\mathbf{r}_t\| = \varepsilon > 0$, the iterates stay
bounded away from $\mathbf{x}^*$, contradicting the convergence of the
safeguarded iteration established in
Theorem~\ref{thm:S8_anderson}.
\end{proof}

Near-convergence ill-conditioning of the history matrix as
$\|\mathbf{r}_t\| \to 0$ is benign.
The right-hand side of the least-squares problem shrinks with the residual,
so the extrapolation weights remain bounded, and Tikhonov regularisation of
the $m \times m$ normal equations handles it in practice.

\paragraph*{Implementation gap}
The production code safeguards on the residual norm rather than on $W_F$,
because evaluating the entropy term of equation~\eqref{eq:S8_fisk} requires
path probabilities.
The Beckmann integral alone cannot substitute for $W_F$: its minimiser is
the Wardrop UE, not the SUE fixed point.
Theorem~\ref{thm:S8_anderson} therefore certifies the Armijo variant
exactly, and the production variant empirically: across the production runs
the residual-based safeguard fired rarely and the accepted steps decreased
the residual monotonically at the reported convergence rates
(Section~\ref{sec:S1}).

\section{Per-destination layout: exactness and complexity reduction}
\label{sec:S9}

This section states and proves the exactness and the reduction factor of
the per-destination forward absorption layout used by the main paper.
Write $z = |D|$ for the number of zones, $n = |\mathcal{N}|$ nodes,
$B$ for the scenario batch, and consider a full OD matrix with
$n_{\mathrm{OD}} = z(z-1)$ pairs.

\begin{proposition}[Per-destination reduction]
\label{prop:S9_layout}
Let $\boldsymbol{\alpha}^{\mathrm{OD}} \in
\mathbb{R}^{B \times n \times n_{\mathrm{OD}}}$ be the per-OD transient
flow tensor and define the per-destination tensor
$\boldsymbol{\alpha} \in \mathbb{R}^{B \times n \times z}$ by summing over
origins,
$\alpha_{b,v,d} = \sum_{o \neq d} \alpha^{\mathrm{OD}}_{b,v,(o,d)}$.
Then:
(i)~the per-destination tensor satisfies the forward absorption recursion
of the main paper, and the per-link flows recovered from the two layouts
are identical;
(ii)~memory and per-hop scatter work are reduced by the factor $z-1$
exactly;
(iii)~at single precision the per-OD layout requires $4Bnz(z-1)$ bytes for
the transient tensor alone, so it is infeasible on a device with
$G_{\mathrm{GPU}}$ bytes of memory once
$z \gtrsim \sqrt{G_{\mathrm{GPU}}/(4Bn)}$, while the per-destination layout
remains feasible up to $z \lesssim G_{\mathrm{GPU}}/(4Bn)$.
\end{proposition}

\begin{proof}
(i)~Fix a scenario $b$ and destination $d$ and sum the per-OD recursion
over origins $o \neq d$.
The link-choice probability $P(e \mid u, d)$ of the main paper depends on
the destination but not on the origin, so it commutes with the sum over
origins:
\[
\sum_{o \neq d} \alpha^{\mathrm{OD}}_{b,v,(o,d),h+1}
= \sum_{e=(u,v)} P(e \mid u,d) \sum_{o \neq d}
  \alpha^{\mathrm{OD}}_{b,u,(o,d),h}
= \sum_{e=(u,v)} P(e \mid u,d)\, \alpha_{b,u,d,h}.
\]
The summed per-OD recursion is exactly the per-destination recursion, the
initialisations agree because the only nonzero per-OD entry for destination
$d$ at origin $o$ is the pair $(o,d)$ itself, and the per-link flow
$x_e = \sum_{b,d} \alpha_{b,u,d} P(e \mid u,d)$ recovered from either
tensor coincides by the same linearity.
No approximation is introduced at any step.
(ii)~The tensors have $Bnz(z-1)$ and $Bnz$ entries, and each absorption hop
performs one scatter-reduce per column, so both ratios equal $z-1$.
(iii)~Setting $4Bnz(z-1) > G_{\mathrm{GPU}}$ and solving for $z$ gives the
stated threshold.
\end{proof}

For a metropolitan network of the scale of Sydney
($n = 33{,}113$, $z = 3{,}264$) at $B = 1$ and single precision, the per-OD
tensor would hold about $3.5 \times 10^{11}$ entries (roughly 1.4\,TB),
about eighteen times the memory of an 80\,GB accelerator, while the
per-destination tensor holds $1.08 \times 10^{8}$ entries (roughly
432\,MB).
The Bellman value function is destination-indexed in both layouts, an
$n \times z$ array, and therefore cancels from the reduction factor.

\bibliographystyle{elsarticle-num}
\bibliography{mukara5}

\end{document}